\documentclass[12pt,draftcls,onecolumn]{ieeetran}
\usepackage{color}
\usepackage{graphicx}
\usepackage{amsmath}
\usepackage{authblk}
\usepackage{times}
\usepackage{dsfont}
\usepackage{subfigure}
\usepackage{cite}
\usepackage{amssymb}
\usepackage{geometry}
\renewcommand{\baselinestretch}{1.8}

\def\x{{\mathbf x}}

\def\t{\mr{T}}

\def\diag{{\rm diag}\,}

\def\Exp{{\mathbb{E}}\,}
\def\tr{{\rm tr}\,}

\def\diag{{\rm diag}\,}

\def\be{\begin{equation}}
\def\ee{\end{equation}}
\def\ba{\left[\begin{array}}
\def\ea{\end{array}\right]}
\def\bea{\begin{eqnarray}}
\def\eea{\end{eqnarray}}

\newcommand{\mb}[1]{\mathbf{#1}}
\newcommand{\mr}[1]{\mathrm{#1}}
\newcommand{\mc}[1]{\mathcal{#1}}
\newcommand{\ol}[1]{\overline{#1}}

\def\ba{{\bf a}}

\def\vec{\mr{vec}}

\newtheorem{theorem}{\textbf{Theorem}}

\newtheorem{lemma}{\textbf{Lemma}}

\IEEEoverridecommandlockouts
\begin{document}
\title{Maximum-rate Transmission
with Improved Diversity Gain for Interference Networks
%\vspace{-1.5cm}
}
%
% Single address.
% ---------------
\date{}
\author{Liangbin Li, Hamid Jafarkhani}
\affil{Center for Pervasive Communications \& Computing, University
of California, Irvine\thanks{ This work was supported in part by the
NSF award CCF-0963925. Part of this work was presented at IEEE
International Symposium on Information Theory (ISIT) 2011. }} %\affil[2]{University of Alberta}
%\ninept
%
\maketitle

\begin{abstract}
Interference alignment (IA) was shown effective for interference
management to improve transmission rate in terms of the degree of
freedom (DoF) gain. On the other hand, orthogonal space-time block
codes (STBCs) were widely used in point-to-point multi-antenna
channels to enhance transmission reliability in terms of the
diversity gain. In this paper, we connect these two ideas, i.e., IA
and space-time block coding, to improve the designs of alignment
precoders for multi-user networks. Specifically, we consider the use
of Alamouti codes for IA because of its rate-one transmission and
achievability of full diversity in point-to-point systems. The
Alamouti codes protect the desired link by introducing orthogonality
between the two symbols in one Alamouti codeword, and create
alignment at the interfering receiver. We show that the proposed
alignment methods can maintain the maximum DoF gain and improve the
ergodic mutual information in \emph{the long-term regime}, while
increasing the diversity gain to $2$ in \emph{the short-term
regime}. The presented examples of interference networks have two
antennas at each node and include the two-user X channel, the
interferring multi-access channel (IMAC), and the interferring
broadcast channel (IBC).

\end{abstract}

\section{Introduction}
Interference plays a major role in open air network communication
and interference management is crucial for future wireless network
designs. Recent research shows much interest in a technique called
interference alignment (IA) that enhances network throughput in
terms of the degree of freedom (DoF) gain (or equivalently the
multiplexing gain). Through the control of either spatial transmit
beamformers \cite{CaJa08,JaSh08,MaTse10} or temporal correlation
patterns\cite{WaGoJa}, interference casts overlapping shadows in the
receive signal space at unintended receivers. Such control minimizes
the dimensions of interference while keeping useful signals
discernable at receivers. The technique is the key to achieve the
maximum DoF gain in interference channels\cite{CaJa08}, X
channels\cite{JaSh08,CaJa09}, and broadcast
channels\cite{WaGoJa,MaTse10} at the cost of simple linear
processing for transmitters and receivers.

In addition to network throughput, reliability in terms of the
diversity gain is another performance metric. When channels are in
deep fading, the signal-to-noise (SNR) level at the receiver is low
and systems cannot support specified transmission rate, which
consequently results in outage events with finite diversity gain.
Various techniques have been intensively studied to improve the
spatial diversity gain, e.g., Alamouti codes\cite{Alamouti98},
space-time block codes (STBCs)\cite{tar99,tar99_2}, and beamforming
methods for point-to-point multi-input multi-output (MIMO) channels;
the interference cancellation (IC) method for multi-access channels
(MACs) \cite{NaSeCa,KaJa}; and the downlink IC method for broadcast
channels (BCs)\cite{Li12}. Conceptually, the DoF gain and the
diversity gain demonstrate different dimensions of performance
metrics in high SNR. The DoF gain reflects the long-term
performance, where systems can have ergodic power constraints (e.g.,
use a Gaussian codebook that has infinite peak power) and
infinite-length channel coding against noise corruption. When the
system has perfect channel state information at the transmitter
(CSIT), rate adaption can be performed with infinite sets of
codebooks. The rate can be instantaneously zero when channels are in
deep fading, or grow linearly with $\log \mr{SNR}$ to boost the
transmission rate\cite{dtse}. A system pursuing the DoF gain
operates in \emph{the long-term regime}. With long-term constraints
on power, decoding delay, and rate, channel outage can be avoided by
choosing a codebook with a rate lower than the instantaneous
capacity. On the other hand, the diversity gain reflects the
short-term performance, where systems have constraints on power,
decoding delay, and rates for a finite number of fading blocks
(e.g., a delay-limited system). With a non-zero minimum rate
constraint, channel outages cannot be avoided and are dominated by
finite diversity gain, although power allocation and rate adaption
can be performed within the constrained blocks\cite{short_report}. A
system pursuing the diversity gain operates in \emph{the short-term
regime}. Both metrics are of equal importance for communication
system designs. We are particularly interested in the spatial
diversity gain, which can be straightforwardly combined with other
forms of diversity, e.g., frequency diversity and time diversity.
The existing alignment methods in \cite{CaJa08, JaSh08}, although
achieve the maximum DoF gain, provide only a spatial diversity gain
of 1 in the short-term regime\cite{Sezgin09}. In this paper, we aim
at improving the diversity gain without losing the maximum DoF gain.

The main idea conceived by STBCs with orthogonal designs is the
orthogonality between embedded symbols\cite{tar99}. The
orthogonality guarantees no SNR loss at the receiver if the
zero-forcing (ZF) method is used to decouple symbols in one block.
The improvement holds for any SNRs. Consequently, full-diversity is
achieved as long as the block code has full-rank. We adopt this idea
into the linear alignment design to protect the desired channels.
While the previous alignment methods only focus on linear IA at
unintended receivers without considering the desired channels, our
proposed method uses STBC to enhance the reliability of desired
channels without affecting alignment at interferring receivers. This
explains the diversity improvement obtained by the proposed methods.
Specifically, since Alamouti code is the only complex orthogonal
design that can achieve rate-one (the maximum possible rate for
orthogonal designs)\cite{tar99}, we embed Alamouti codes into
alignment designs. Alamouti code also has another nice property that
its $2\times 2$ matrix structure is closed under matrix
multiplication and addition. This property is utilized for the IC
method in MACs such that the Alamouti structure of the equivalent
channel matrix is preserved after cancelling the interfering
users\cite{NaSeCa}. Enlightened by these facts, we propose new
alignment methods using Alamouti codes.

We motivate the idea in a double-antenna $2\times 2$ X channel,
where two transmitters send symbols to each of the two receivers.
The maximum DoF gain of such a network is known to be
$\frac{4}{3}\times 2=\frac{8}{3}$\cite{JaSh08}, achievable by symbol
extensions over three channel uses and sending two symbols over each
communication direction. Since each transmitter has two antennas and
only two symbols are sent to each receiver, we propose to convey
these two symbols in a block with Alamouti structure. Alignment at
interferring receivers is achieved on an equivalent channel matrix
with Alamouti structure. Therefore, the two symbols of the same user
are orthogonal to each other and decoupling them does not incur SNR
loss. Consequently, the maximum transmit diversity gain is obtained.
The contributions of this paper are summarized as
\begin{enumerate}
  \item In the two-user double-antenna X channel,
  compared to the linear alignment method in \cite{JaSh08},
  our proposed scheme achieves higher diversity gain, i.e., a
  diversity gain of $2$ at the same DoF gain $\frac{8}{3}$. Our proposed method only requires local CSIT instead of global
  CSIT as assumed in \cite{JaSh08}. In other words, each transmitter only needs to know the channel
  information from itself to both receivers.
  \item The proposed method can be extended with the same
diversity gain improvement to
  cellular networks such as the interferring MAC (IMAC) and the interferring BC (IBC) \cite{Suh08}, where inter-cell interference affects desired communication. The mobile stations
(MSs) only require local channel information. Since the IMAC and the
IBC are dual to each other, we use the idea of duality\cite{ViTse03,
ViJinGo03,Li12} to transform the alignment solution in the IMAC to
the solution in the IBC. Simulation shows significant bit error rate
(BER) performance improvement compared to the downlink IA method
\cite{Suh11}.
\item Improvements are not limited to the diversity gain in the high SNR regime. Our proposed method also demonstrates improvements, compared to the aforementioned existing methods in the literature,  on the
achievable ergodic mutual information at any SNR.
\end{enumerate}

IA with diversity benefits is also parallelly studied in
\cite{Ning11,Song12,Feng12} at rate-one (one DoF is communicated per
node pair) for interference channels and X channels. Notably,
\cite{Ning11} considers feasibility of IA for diversity gain in
interference channels. Besides interference alignment at unintended
receivers, transmit beamformers are also designed to maximize the
signal to interference-plus-noise (SINR). Consequently, their
designs for a three-user interference channel with three antennas at
transmitters and two antennas at receivers bring a diversity gain of
$3$. Note that our paper differs from \cite{Ning11, Song12} in the
number of DoFs transmitted per node pair. We allow the network to
achieve the maximum DoF gain, while in \cite{Ning11, Song12}, each
transmitter sends only one DoF to the intended receiver. Naturally,
it is more challenging to design a system transmitting more DoFs.
Secondly, our system allows symbol extensions or multiple channel
uses, while their system does not use symbol extensions. Thirdly,
the mechanisms of the protection for the desired link are different.
Our paper considers STBCs, while their papers use transmit
beamformers.

The rest of the paper is organized as follows. Section
\ref{Sec-Model} discusses the channel model and reviews the
alignment scheme in \cite{JaSh08}. In Section \ref{Sec-X}, we
present the alignment method using Alamouti designs for X channels.
Section \ref{Sec-Cellular} extends the proposed method to the IMAC
and IBC. Simulations are shown in Section \ref{Sec-Simulation} and
conclusions are given in Section \ref{Sec-Conclusion}. Proofs of
theorems are provided in the appendices.

\textbf{Notations:} Let a vector $\mb{a}\in \mathds{C}^{N\times 1}$
be drawn from a complex vector space with dimension $N\times 1$. We
denote $\diag(\mb{a})\in \mathds{C}^{N\times N}$ as a diagonal
matrix whose diagonal entries are copied from the entries in
$\mb{a}$. For a matrix $\mb{A}$, we use $\mb{A}^\t$, $\mb{A}^*$,
$\tr(\mb{A})$, $\vec(\mb{A})$, and $\|\mb{A}\|$ to denote its
transpose, Hermitian, trace, vectorization, and Frobenius norm,
respectively. For two matrices $\mb{A}_1$ and $\mb{A}_2$, the
notations $\mb{A}_1\otimes \mb{A}_2$ are used for the Kronecker
product. When matrices $\mb{A}_1, \mb{A}_2\in \mathds{C}^{N\times
N}$ are drawn from the same matrix space, we use $\mb{A}_1\prec
\mb{A}_2$ to denote their difference $\mb{A}_2-\mb{A}_1$ to be
positive definite. The notation $\mc{CN}(0,1)$ is used for a
circular symmetric complex Gaussian distribution with zero mean and
variance 1.

\section{Previous Linear Alignment in Two-user X channels}\label{Sec-Model}
This section explains the X channel model and the previous linear
alignment solution for X channels. Consider an $M$-antenna $2\times
2$ MIMO X channel. Two transmitters send symbols to two receivers,
where each node is equipped with $M$ antennas. Each of the two
transmitters has $K$ independent symbols intended for each of the
two receivers. In other words, Transmitter $j$ has symbol
$s^{[ji]}_k$ for Receiver $i$, where $j,i \in \{1,2\}, k\in
\{1,2,\ldots, K\}$. Throughout the paper, we use indices $j,i,k$ for
transmitter, receiver, and symbol, respectively. The expected power
of $s^{[ji]}_k$ is $\Exp\left|s^{[ji]}_k\right|^2=P$, where $P$ is
the available power at the transmitter per channel use. When the
system is operated in the long-term regime, a Gaussian codebook can
be used for $s^{[ji]}_k$ and each symbol carries one DoF gain. In
other words, the bit rate of $s^{[ji]}_k$ scales like $\log P$ in
the high SNR regime. Since each symbol carries one DoF gain, symbol
rate is equal to the DoF gain. We call a transmission method that
achieves the maximum DoF gain \emph{a maximum-rate scheme}. In the
short-term regime, $s^{[ji]}_k$ is generated from a finite set of
codebooks. With a non-zero minimum rate constraint, system
performance is dominated by the worst codebook. Without loss of
generality, we can assume $s^{[ji]}_k$ is uncoded and drawn from
fixed constellations with finite cardinality, e.g., QPSK or 16QAM.
Denote the constellation as $\mc{S}$ and its cardinality as
$\left|\mc{S}\right|$. The bit rate of $s^{[ji]}_k$ is fixed to
$\log \left|\mc{S}\right|$ at any SNR. For simplicity, we will
present the paper by assuming fixed constellation for $s^{[ji]}_k$
to study the achievable diversity gain unless otherwise stated.

To focus on spatial diversity gain, we model channels as Rayleigh
block fading. The channel matrix from Transmitter $j$ to Receiver
$i$ is denoted as $\mb{H}^{[ji]}\in \mathds{C}^{M\times M}$. Then,
the $(m,n)$th entry in $\mb{H}^{[ji]}$, denoted as $h^{[ji]}_{mn}$,
is the fading channel coefficient from transmit Antenna $m$ to
receive Antenna $n$. We model $h^{[ji]}_{mn}$ as drawn from
i.~i.~d.~$\mc{CN}(0,1)$ distribution. In addition, all channels are
assumed block fading (also known as constant channels), i.e., all
channels keep unchanged during the transmission. Let the transmit
duration be $T$ channel uses, and Transmitter $j$ embeds $2K$
symbols, i.e., $s^{[j1]}_k$ and $s^{[j2]}_k$, into a block
$\mb{X}^{[j]}\in \mathds{C}^{T\times M}$. The signal block sampled
at Receiver $i$ can be written as
\begin{align}\label{eq-channeleq}
\mb{Y}^{[i]}=\mb{X}^{[1]}\mb{H}^{[1i]}+\mb{X}^{[2]}\mb{H}^{[2i]}+\mb{W}^{[i]},\
i\in\{1,2\}.
\end{align}
where $\mb{Y}^{[i]}, \mb{W}^{[i]}\in \mathds{C}^{T\times M}$ and
$\mb{W}^{[i]}$ denotes the additive white Gaussian noise (AWGN)
matrix at Receiver $i$. Each entry in $\mb{W}^{[i]}$ has
i.~i.~d.~$\mc{CN}(0,1)$ distribution.

The reason for choosing $2\times 2$ MIMO X channels is for its
simplicity and the existence of linear alignment using finite
signaling dimensions. For a general $J\times R$ X channels with
$\min\{J,R\}>2$, the feasibility of linear IA is still open, and so
far the best achievable solution is the asymptotical alignment that
requires infinite signaling dimensions to approach the maximum DoF
gain \cite{CaJa09}.

In what follows, we review the linear IA method in \cite{JaSh08} for
the $2\times 2$ MIMO X channels with a change of notations used in
this paper. The alignment achieves the maximum symbol rate of
$\frac{4M}{3}$ symbols/channel use over the network. The design
needs three channel uses for signaling, i.e., $T=3$. Transmitter $j$
linearly combines $2M$ symbols ($M$ symbols for each receiver) into
the transmitted block $\mb{X}^{[j]}$. In total, $4M$ symbols are
transmitted over the network in $3$ channels uses, which provides a
symbol rate of $\frac{4M}{3}$ symbols/channel use. The design is
based on the vector transform of system equation in
\eqref{eq-channeleq},
\begin{align}\label{eq-columneq}
\underset{\mb{y}^{[i]}}{\underbrace{\vec\left(
\mb{Y}^{[i]}\right)}}=\underset{\ol{\mb{H}}^{[1i]}}{\underbrace{\left(\mb{H}^{[1i]\t}\otimes\mb{I}_3\right)}}
\underset{\mb{x}^{[1]}}{\underbrace{\vec\left(\mb{X}^{[1]}\right)}}+\underset{\ol{\mb{H}}^{[2i]}}{\underbrace{\left(\mb{H}^{[2i]\t}\otimes\mb{I}_3\right)}}\underset{\mb{x}^{[2]}}{\underbrace{\vec\left(\mb{X}^{[2]}\right)}}+\underset{\mb{w}^{[i]}}{\underbrace{\vec\left(\mb{W}^{[i]}\right)}},\
i\in\{1,2\},
\end{align}
where $\mb{y}^{[i]},\mb{x}^{[j]}, \mb{w}^{[i]}\in
\mathds{C}^{3M\times 1}$ and $\ol{\mb{H}}^{[ji]}\in
\mathds{C}^{3M\times 3M}$. The equivalent transmitted vector
$\mb{x}^{[j]}$ is designed as linear beamforming of symbols
$s^{[ji]}_k$
\begin{align}\mb{x}^{[j]}=\ol{\mb{v}}^{[j1]}\left[s^{[j1]}_1\
s^{[j1]}_2\ \ldots \
s^{[j1]}_M\right]^{\t}+\ol{\mb{v}}^{[j2]}\left[s^{[j2]}_1\
s^{[j2]}_2\ \ldots \ s^{[j2]}_M\right]^{\t} ,\label{eq-prevbf}
\end{align}
where $\ol{\mb{v}}^{[ji]}$ denotes the $3M\times M$ beamforming
matrix from Transmitter $j$ to Receiver $i$. The symbols
$s^{[11]}_k, (k=1,2,\ldots, M)$ are intended for Receiver 1, hence
become interference for Receiver 2. The beamformer
$\ol{\mb{v}}^{[11]}$ aligns $s^{[11]}_k$ with $s^{[21]}_k$ in an
$M$-dimensional subspace at Receiver 2 as
\begin{align*}
\ol{\mb{H}}^{[22]}\ol{\mb{v}}^{[21]}=\ol{\mb{H}}^{[12]}\ol{\mb{v}}^{[11]}.\end{align*}
Similarly, the symbols $s^{[12]}_k$ are aligned with $s^{[22]}_k$ in
an $M$-dimensional subspace at Receiver 1 as
\begin{align*}\ol{\mb{H}}^{[21]}\ol{\mb{v}}^{[22]}=\ol{\mb{H}}^{[11]}\ol{\mb{v}}^{[12]}.
\end{align*}
Since channel matrices are almost surely full rank, we can
immediately obtain $\ol{\mb{v}}^{[21]}$ and $\ol{\mb{v}}^{[22]}$ as
functions of $\ol{\mb{v}}^{[11]}$ and $\ol{\mb{v}}^{[12]}$,
respectively,
\begin{align}\label{eq-bfstep1}
\ol{\mb{v}}^{[21]}=\left(\ol{\mb{H}}^{[22]}\right)^{-1}\ol{\mb{H}}^{[12]}\ol{\mb{v}}^{[11]},\ol{\mb{v}}^{[22]}=\left(\ol{\mb{H}}^{[21]}\right)^{-1}\ol{\mb{H}}^{[11]}\ol{\mb{v}}^{[12]}.\end{align}
The remaining beamformers $\ol{\mb{v}}^{[11]}$ and
$\ol{\mb{v}}^{[12]}$ are designed for linear independence between
the desired signal space and the interference subspace as
\begin{align}\label{eq-bfstep2}\ol{\mb{v}}^{[11]}=\mb{U}(\mb{I}_M\otimes
\mb{E}_1),\ \ol{\mb{v}}^{[12]}=\mb{U}(\mb{I}_M\otimes \mb{E}_2),
\end{align} where $\mb{E}_1=\left[1,1,0\right]^{\t}$,
$\mb{E}_2=\left[1,0,1\right]^{\t}$, and $\mb{U}\in
\mathds{C}^{3M\times 3M}$ is denoted as the eigenvector matrix of
$\left(\ol{\mb{H}}^{[11]}\right)^{-1}\ol{\mb{H}}^{[21]}\left(\ol{\mb{H}}^{[22]}\right)^{-1}\ol{\mb{H}}^{[12]}$
whose eigenvalues are arranged as $\lambda_{1}\neq \lambda_{2},
\lambda_{1}\neq \lambda_{3}, \lambda_{4}\neq \lambda_{5},
\lambda_{4}\neq \lambda_{6}, \ldots, \lambda_{3M-2}\neq
\lambda_{3M-1}, \lambda_{3M-2}\neq \lambda_{3M}$. At each receiver,
ZF is performed to cancel interference and separate useful symbols
to obtain symbol-by-symbol decodings. From \eqref{eq-bfstep1} and
\eqref{eq-bfstep2}, each transmitter requires global channel
information to design the beamformers. For simplicity, we call this
transmission method \emph{the JaSh scheme}.

\section{Alamouti-coded Transmission for X channels}\label{Sec-X}
In this section, we present how Alamouti designs can be used for the
linear IA in X channels. While previous alignment schemes consider
the designs of alignment precoders only based on interfering links
and disregard the desired links, we incorporate the idea of Alamouti
designs to protect the transmission of desired symbols, because
Alamouti codes achieve full transmit spatial diversity in
point-to-point MIMO systems\cite{Alamouti98}. Consequently, the
proposed alignment method can achieve the same maximum symbol-rate
as the scheme in \cite{JaSh08} but with a higher diversity gain. To
use Alamouti codes, we assume each node in the X channel has two
antennas, i.e., $M=2$. We first present the transmission method in
Subsection \ref{subsec-2ant}, then analyze the achievable diversity
gain in Subsection \ref{subsec-analysis}. In this section, we assume
that each transmitter has channel information from itself to both
receivers, i.e., Transmitter $j$ only knows $\mb{H}^{[j1]}$ and
$\mb{H}^{[j2]}$. Receivers require global channel information.

\subsection{The transmission method}\label{subsec-2ant}
The maximum rate of the double-antenna $2\times 2$ X channel is
$2\times \frac{4}{3}=\frac{8}{3}$\cite{JaSh08}. To achieve this
rate, we design each transmitter to send two symbols to each of the
two receivers in three channel uses, i.e., $K=2$ and $T=3$. The
system diagram is shown in Fig.~\ref{fig-diag}. The transmitted
block $\mb{X}^{[j]}$ is designed as
\begin{align}\label{eq-newbeamformer}\mb{X}^{[j]}=\sqrt{\frac{3}{4}}\left(\left[\begin{array}{cc}s^{[j1]}_1&s^{[j1]}_2\\
-s^{[j1]*}_{2} & s^{[j1]*}_{1}\\0 &
0\end{array}\right]\mb{V}^{[j1]}+\left[\begin{array}{cc}0 &
0\\-s^{[j2]*}_{2} & s^{[j2]*}_{1}\\s^{[j2]}_1&s^{[j2]}_2
\end{array}\right]\mb{V}^{[j2]}\right),\ j \in \{1,2\},\end{align}
where $\mb{V}^{[ji]}\in \mathds{C}^{2\times 2}$ denotes the
beamforming matrix from Transmitter $j$ to Receiver $i$. Recall that
from \eqref{eq-channeleq}, the vertical and horizontal dimensions of
$\mb{X}^{[j]}$ represent temporal and spatial dimensions,
respectively. The symbols to Receiver $1$ are encoded by Alamouti
designs and transmitted in the first two time slots; whereas the
symbols to Receiver $2$ are encoded by Alamouti designs too, but
transmitted in the last two time slots. Compared to the designs in
\eqref{eq-prevbf}, our scheme allows each transmitter to send linear
combinations of both the original symbols and their conjugate. The
beamforming matrices are designed to align $s^{[11]}_k$ and
$s^{[21]}_k$ at Receiver 2, and align $s^{[12]}_k$ and $s^{[22]}_k$
at Receiver 1 as shown in Fig.~\ref{fig-space}. Specifically, we
design the beamforming matrix as the normalized inversion of the
cross channel matrix,
\begin{align}\label{eq-beamformer}
\mb{V}^{[ji]}=c^{[ji]}\left(\mb{H}^{[j\bar{i}]}\right)^{-1},\
j,i\in\{1,2\} \end{align} where the index $\bar{i}$ denotes the
receiver other than Receiver $i$ and the coefficient
$c^{[ji]}=1/\left\|\left(\mb{H}^{[j\bar{i}]}\right)^{-1}\right\|$ is
to satisfy the power constraint\footnote{In this paper, we design
power to be equally allocated between symbols for two users, because
we focus on the diversity gain performance. Further power allocation
to maximize the array gain is possible. }
$\tr\left(\mb{V}^{[ji]}\mb{V}^{[ji]*}\right)=1$. This power
constraint implicitly ensures each entry in $\mb{V}^{[ji]}$ to be
smaller than 1 and avoids high peak powers. The coefficient
$\sqrt{\frac{3}{4}}$ in \eqref{eq-newbeamformer} is to normalize the
transmit power to $\underset{s_k^{[ji]}}{\Exp}
\tr\left(\mb{X}^{[j]}(\mb{X}^{[j]})^*\right)=3P$ in three channel
uses. Inserting \eqref{eq-newbeamformer} into \eqref{eq-channeleq},
the receive signal blocks can be expanded as
{\setlength{\arraycolsep}{2pt}\begin{align}\label{eq-syst1}
  &\mb{Y}^{[1]}=\sum_{j\in \{1,2\}}\sqrt{\frac{3}{4}}\left[\begin{array}{cc}s^{[j1]}_1&s^{[j1]}_2\\ -s^{[j1]*}_{2} & s^{[j1]*}_{1}\\0 & 0\end{array}\right]\tilde{\mb{H}}^{[j1]}
  +\sqrt{\frac{3}{4}}\left[\begin{array}{cc}0 & 0\\-c^{[11]}s^{[12]*}_{2}-c^{[21]}s^{[22]*}_{2} &
  c^{[11]}s^{[12]*}_{1}+c^{[21]}s^{[22]*}_{1}\\c^{[11]}s^{[12]}_1+c^{[21]}s^{[22]}_1&c^{[11]}s^{[12]}_2+c^{[21]}s^{[22]}_2
  \end{array}\right]+\mb{W}^{[1]},\\ \label{eq-syst2}
  &\mb{Y}^{[2]}=\sum_{j\in \{1,2\}}\sqrt{\frac{3}{4}}\left[\begin{array}{cc}0 & 0\\ -s^{[j2]*}_{2} & s^{[j2]*}_{1}\\s^{[j2]}_1&s^{[j2]}_2\end{array}\right]\tilde{\mb{H}}^{[j2]}
  +\sqrt{\frac{3}{4}}\left[\begin{array}{cc}c^{[12]}s^{[11]}_1+c^{[22]}s^{[21]}_1&c^{[12]}s^{[11]}_2+c^{[22]}s^{[21]}_2\\-c^{[12]}s^{[11]*}_{2}-c^{[22]}s^{[21]*}_{2} &
  c^{[12]}s^{[11]*}_{1}+c^{[22]}s^{[21]*}_{1}\\0 & 0
  \end{array}\right]+\mb{W}^{[2]},
\end{align}}
\hspace{-3pt}where
$\tilde{\mb{H}}^{[ji]}=\mb{V}^{[ji]}\mb{H}^{[ji]}$ denotes the
equivalent channels that incorporate beamforming matrices. In the
above equations, the first term represents desired symbols, whereas
the second term represents interference. It can be observed that the
interference term still has Alamouti structure, since $c^{[ji]}$ is
a real number. In other words, $s^{[12]}_k$ and $s^{[22]}_k$ are
aligned at Receiver 1, while $s^{[11]}_k$ and $s^{[21]}_k$ are
aligned at Receiver 2. We can further convert the system equations
into vector forms to study the receive signal space. Let us denote
the $t$th row of $\mb{Y}^{[i]}$ and $\mb{W}^{[i]}$ be
$\mb{y}_{t}^{[i]}$ and $\mb{w}_{t}^{[i]}$, respectively, where $t\in
\{1,2,3\}$. Denote the aligned interfering symbols as
$I_k^{[1]}=c^{[11]}s^{[12]}_k+c^{[21]}s^{[22]}_k,I_k^{[2]}=c^{[12]}s^{[11]}_k+c^{[22]}s^{[21]}_k$,
and the $(m,n)$th entry of $\tilde{\mb{H}}^{[ji]}$ as
$\tilde{{h}}^{[ji]}_{mn}$. The receiver calculates
$\tilde{\mb{y}}^{[i]}=\vec\left(\left[\mb{y}_{1}^{[i]*},(-1)^{i}\left(\mb{y}_{2}^{[i]}\right)^{\t},\mb{y}_{3}^{[i]*}\right]^*\right)$,
and Eqns.~\eqref{eq-syst1} and \eqref{eq-syst2} can be converted as
\begin{align}\label{eq-vecsys1}
&\tilde{\mb{y}}^{[1]}=\sqrt{\frac{3}{4}}
\left[\begin{array}{cccccc}\tilde{h}^{[11]}_{11}&\tilde{h}^{[11]}_{21}&\tilde{h}^{[21]}_{11}&\tilde{h}^{[21]}_{21}&0&0\\ -\tilde{h}^{[11]*}_{21}&\tilde{h}^{[11]*}_{11}&-\tilde{h}^{[21]*}_{21}&\tilde{h}^{[21]*}_{11}&0&1\\0&0&0&0&1&0\\
\tilde{h}^{[11]}_{12}&\tilde{h}^{[11]}_{22}&\tilde{h}^{[21]}_{12}&\tilde{h}^{[21]}_{22}&0&0\\
-\tilde{h}^{[11]*}_{22}&\tilde{h}^{[11]*}_{12}&-\tilde{h}^{[21]*}_{22}&\tilde{h}^{[21]*}_{12}&-1&0\\0&0&0&0&0&1\end{array}\right]\left[\begin{array}{c}s^{[11]}_1\\s^{[11]}_2\\s^{[21]}_1\\s^{[21]}_2\\I^{[1]}_1\\I_2^{[1]}\end{array}\right]
+\tilde{\mb{w}}^{[1]}\end{align} at Receiver $1$, and
\begin{align}\label{eq-vecsys2}
&\tilde{\mb{y}}^{[2]}=\sqrt{\frac{3}{4}}\left[\begin{array}{cccccc}0&0&0&0&1&0\\
\tilde{h}^{[12]*}_{21}&-\tilde{h}^{[12]*}_{11}&\tilde{h}^{[22]*}_{21}&-\tilde{h}^{[22]*}_{11}&0&-1\\
\tilde{h}^{[12]}_{11}&\tilde{h}^{[12]}_{21}&\tilde{h}^{[22]}_{11}&\tilde{h}^{[22]}_{21}&0&0\\0&0&0&0&0&1\\ \tilde{h}^{[12]*}_{22}&-\tilde{h}^{[12]*}_{12}&\tilde{h}^{[22]*}_{22}&-\tilde{h}^{[22]*}_{12}&1&0\\
\tilde{h}^{[12]}_{12}&\tilde{h}^{[12]}_{22}&\tilde{h}^{[22]}_{12}&\tilde{h}^{[22]}_{22}&0&0\end{array}\right]\left[\begin{array}{c}s^{[12]}_2\\s^{[12]}_1\\s^{[22]}_1\\s^{[22]}_2\\I_1^{[2]}\\I_2^{[2]}\end{array}\right]
+\tilde{\mb{w}}^{[2]}
\end{align}
at Receiver $2$, where $\tilde{\mb{y}}^{[i]},
\tilde{\mb{w}}^{[i]}\in \mathds{C}^{6\times 1}$ and
$\tilde{\mb{w}}^{[i]}=\vec \left( \left[ \mb{w}_{1}^{[i]*},
(-1)^i\mb{w}_{2}^{[i]\t}, \mb{w}_{3}^{[i]*}\right]^*\right)$ denotes
the equivalent AWGN vector at Receiver $i$. It can be observed that
the equivalent channel vectors of $s_1^{[ji]}$ and $s_2^{[ji]}$
(correspond to the $(2j-1)$ and $(2j)$th columns in the equivalent
channel matrix) are orthogonal. Thus, the desired links are enhanced
by embedding Alamouti codes into alignment. The receive signal space
is illustrated in Fig.~\ref{fig-space}.

In what follows, we explain receiver decoding using IC originally
proposed for MAC\cite{NaSeCa}. Although IC is essentially ZF, IC
avoids high dimensional matrix processing (simplify the computation
of matrix inversion in the projection matrix). Since the designs of
the network is symmetric to each receiver, we focus only on the
processing at Receiver 1 to simplify presentation. Processing at
Receiver 2 is similar and has the same performance as that of
Receiver 1. Since we only discuss Receiver 1, in what follows, we
will remove receiver index $i$ from $\tilde{\mb{y}}^{[i]}$ and
$\tilde{\mb{w}}^{[i]}$ to simplify the presentation. The IC has the
following two steps:
\subsubsection{Step 1: Remove aligned interference}
Let the $\tau$th entry of $\tilde{\mb{y}}$ and $\tilde{\mb{w}}$ in
\eqref{eq-vecsys1} be $\tilde{y}_{\tau}$ and $\tilde{w}_{\tau}$,
respectively. Since the equivalent channels for interference
$I_1^{[1]}$ and $I_2^{[1]}$ are constant in \eqref{eq-vecsys1}, the
aligned interference $I_1^{[1]}$ and $I_2^{[1]}$ can be cancelled by
\begin{align}\label{eq-IC}\left[\begin{array}{cccc}\tilde{y}_{1}&
\tilde{y}_{2}+\tilde{y}_{6}& \tilde{y}_{4}&
\tilde{y}_{5}-\tilde{y}_{3}\end{array}\right]^{\t}.\end{align}
Let
{\setlength{\arraycolsep}{2pt}$\hat{\mb{y}}_1=\left[\begin{array}{cc}\tilde{y}_{1}&
\tilde{y}_{2}+\tilde{y}_{6}\end{array}\right]^{\t},
\hat{\mb{y}}_2=\left[\begin{array}{cc}\tilde{y}_{4}&
\tilde{y}_{5}-\tilde{y}_{3}\end{array}\right]^{\t},
\hat{\mb{w}}_1=\left[\begin{array}{cc}\tilde{w}_{1}&
\tilde{w}_{2}+\tilde{w}_{6}\end{array}\right]^{\t},\
\hat{\mb{w}}_2=\left[\begin{array}{cc}\tilde{w}_{4}&
\tilde{w}_{5}-\tilde{w}_{3}\end{array}\right]^{\t}$}. The
resulting equivalent system equation can be simplified as
\begin{align}\label{eq-system1}
\left[\begin{array}{c}\hat{\mb{y}}_1\\\hat{\mb{y}}_2\end{array}\right]=\sqrt{\frac{3}{4}}\left(\left[\begin{array}{c}\hat{\mb{H}}^{[11]}_1\\\hat{\mb{H}}^{[11]}_2\end{array}\right]\left[\begin{array}{c}s^{[11]}_1\\s^{[11]}_2\end{array}\right]+\left[\begin{array}{c}\hat{\mb{H}}^{[21]}_1\\\hat{\mb{H}}^{[21]}_2\end{array}\right]\left[\begin{array}{c}s^{[21]}_1\\s^{[21]}_2\end{array}\right]\right)+\left[\begin{array}{c}\hat{\mb{w}}_1\\\hat{\mb{w}}_2\end{array}\right],
\end{align}
where $\hat{\mb{H}}^{[j1]}_n\in \mathds{C}^{2\times 2}$ has an
Alamouti structure
\begin{align}\nonumber
\hat{\mb{H}}^{[j1]}_n=\left[\begin{array}{cc}\tilde{h}_{1n}^{[j1]}&\tilde{h}^{[j1]}_{2n}\\
\tilde{h}_{2n}^{[j1]*}&-\tilde{h}_{1n}^{[j1]*}\end{array}\right],\
j\in\{1,2\}.
\end{align}
\subsubsection{Step 2: Decouple symbols from different transmitters}
The system equation in \eqref{eq-system1} is similar to that of a
MAC system with two double-antenna transmitters and one
double-antenna receiver. The equivalent noise vector
$\left[\hat{\mb{w}}_1^{\t}\ \hat{\mb{w}}_2^{\t}\right]^{\t}$ is
white but does not have identical variances for each entry. IC is
applicable to decouple $s^{[11]}_1$ and $s^{[11]}_2$ from
$s^{[21]}_1$ and $s^{[21]}_2$. Receiver 1 conducts
\begin{align}\label{eq-remainedeq}
\underset{\hat{\mb{y}}}{\underbrace{\frac{\hat{\mb{H}}_1^{[21]*}}{\left\|\hat{\mb{H}}_1^{[21]}\right\|^2}\hat{\mb{y}}_1-\frac{\hat{\mb{H}}_2^{[21]*}}{\left\|\hat{\mb{H}}_2^{[21]}\right\|^2}\hat{\mb{y}}_2}}=&\sqrt{\frac{3}{4}}\underset{\hat{\mb{H}}}{\underbrace{\left(\frac{\hat{\mb{H}}_1^{[21]*}\hat{\mb{H}}_1^{[11]}}{\left\|\hat{\mb{H}}_1^{[21]}\right\|^2}-\frac{\hat{\mb{H}}_2^{[21]*}\hat{\mb{H}}_2^{[11]}}{\left\|\hat{\mb{H}}_2^{[21]}\right\|^2}\right)}}\left[\begin{array}{c}s^{[11]}_1\\s^{[11]}_2\end{array}\right]+\frac{\hat{\mb{H}}_1^{[21]*}\hat{\mb{w}}_1}{\left\|\hat{\mb{H}}_1^{[21]}\right\|^2}-\frac{\hat{\mb{H}}_2^{[21]*}\hat{\mb{w}}_2}{\left\|\hat{\mb{H}}_2^{[21]}\right\|^2}.
\end{align}
Due to the completeness of matrix addition, matrix
multiplication, and scalar multiplication of the Alamouti
matrix, the equivalent channel matrix $\hat{\mb{H}}$ still has
the Alamouti structure. Thus, $s^{[11]}_k$ can be decoded by
\begin{align}\label{eq-MLdecoding}
s^{[11]}_k=\arg\max_{s}\hat{\mb{h}}_k^*\hat{\mb{y}}s,
k\in\{1,2\},
\end{align} where $\hat{\mb{h}}_k$ denotes the $k$th column of
$\hat{\mb{H}}$. Note that the decoding complexity is
symbol-by-symbol. Similar to \eqref{eq-remainedeq}, we can decouple
$s^{[21]}_1$ and $s^{[21]}_2$ by calculating
$\frac{\hat{\mb{H}}_1^{[11]*}}{\left\|\hat{\mb{H}}_1^{[11]}\right\|^2}\hat{\mb{y}}_1-\frac{\hat{\mb{H}}_2^{[11]*}}{\left\|\hat{\mb{H}}_2^{[11]}\right\|^2}\hat{\mb{y}}_2$.
Similar operations can be performed at Receiver 2 to decode
$s^{[12]}_k$ and $s^{[22]}_k$. Therefore, four procedures of
symbol-by-symbol decoding are required at each receiver to recover
desired symbols.

\subsection{Performance analysis}\label{subsec-analysis}
This subsection provides diversity gain analysis in the short-term
regime. Further, we show that the proposed scheme does not lose the
DoF gain in the long-term regime.

In a point-to-point channel, diversity gain is defined as the
asymptotical slope of BER with respect to the receive SNR in the
high SNR regime. For our considered network model, we define
diversity gain as the asymptotical rate of BER with respect to power
$P$ for the symbol-by-symbol decoding given in
\eqref{eq-MLdecoding}. A diversity calculation technique using
instantaneous normalized receive SNR was proposed in
\cite{divthm_report} for short-term communication systems. For a
vector channel with an equivalent system equation
$\mb{y}=\mb{h}s+\mb{w}$, where $\mb{y}, \mb{h}, s, \mb{w}$ denote
the receive signal vector, the equivalent channel vector, transmit
symbol, and the equivalent noise vector, respectively. The
\emph{instantaneous normalized receive SNR} for symbol $s$ is
defined as $\gamma=\mb{h}^*\mb{\Sigma}^{-1}\mb{h}$, where
$\mb{\Sigma}$ is the covariance matrix of $\mb{w}$. Diversity gain
for the maximum-likelihood (ML) decoding of this equivalent system
equation can be calculated as
\begin{align}\label{eq-div}
d=-\lim_{\epsilon\rightarrow 0}\frac{\log P(\gamma<\epsilon)}{\log
\epsilon},\end{align} where $P(\gamma<\epsilon)$ denotes the outage
probability of $\gamma$. Using this technique, we present the
following theorems.

\begin{theorem}\label{thm-prediv}
In the short-term regime, the JaSh scheme achieves a diversity gain
no more than $1$ for the $2\times 2$ double-antenna X channel.
\end{theorem}
\begin{proof}
See Appendix \ref{appenpre} for proof.
\end{proof}

The intuition of the theorem can be explained as follows. The
receiver observes a six-dimensional signal space, in which two
dimensions are for aligned interference and four dimensions are for
desired symbols. The equivalent channel vectors for desired symbols
are randomly distributed in the receive signal space as shown from
\eqref{eq-sys1} (the beamforming vectors $\mb{u}_1, \mb{u}_2$ depend
on channels of $\mb{H}^{[21]},\mb{H}^{[12]}, \mb{H}^{[22]}$, while
the equivalent channel matrix is $\mb{H}^{[11]}$ for all desired
symbols). By a ZF receiver, the projection to cancel the aligned
interference and decouple the desired symbols incurs SNR loss. Thus,
the resulting diversity gain is 1.

\begin{theorem}\label{thm-div}
In the short-term regime, the proposed alignment method with
Alamouti designs achieves a diversity gain of 2 for the $2\times 2$
double-antenna X channel.
\end{theorem}
\begin{proof}
See Appendix \ref{appen1} for proof.
\end{proof}

This diversity improvement can be intuitively explained as follows.
Compared to the JaSh scheme, two desired symbols are orthogonal (See
\eqref{eq-vecsys1} and \eqref{eq-vecsys2}) due to the use of
Alamouti structure at transmitters. After removing aligned
interference, the system equation in \eqref{eq-system1} is similar
to a MAC with two double-antenna transmitters and one double-antenna
receiver. The IC uses one receive antenna to decouple symbols. Then,
the receive diversity of the proposed scheme is 1. A transmit
diversity gain of 2 is achievable through Alamouti designs. The
total diversity gain is the product of the transmit diversity and
receiver diversity, i.e., it is 2.

Next, we discuss the proposed scheme in the long-term regime. In
this case, Gaussian codebooks can be used for $s_k^{[ji]}$, and
transmitters adjust the rate over infinite sets of Gaussian
codebooks based on CSIT and the transmit power. The network can
reliably transmit information without outage assuming infinite
coding. The DoF gain is defined as the asymptotical ratio between
the bit-rate and $\log P$\cite{dtse}. For our proposed scheme, each
symbol $s_k^{[ji]}$ can be viewed as a data stream whose bit-rate
can be adjusted adaptively. The achievable DoF gain is shown in the
following theorem.

\begin{theorem}\label{thm-div}
In the long-term regime, the proposed alignment method with Alamouti
designs achieves the maximum DoF gain of $\frac{8}{3}$ for the
$2\times 2$ double-antenna X channel.
\end{theorem}
\begin{proof}
Since interferring symbols are aligned by the design and appear in
different temporal dimensions compared to the desired symbol (At
Receiver 1, the desired symbols are received in time Slots 1 and 2,
and interferring symbols are received in time Slots 2 and 3), the
desired symbols can be decoupled from the interferring symbols.
Then, it is sufficient to show the linear independence among the
desired symbols. We only show the linear independence at Receiver 1,
since the channel matrix of the four desired symbols has the same
structure at Receiver 2. We need to prove that the following
$4\times 4$ matrix has full rank
\begin{align}\label{eq-matrix}
\left[\begin{array}{cccc}\tilde{h}^{[11]}_{11}&\tilde{h}^{[11]}_{21}&\tilde{h}^{[21]}_{11}&\tilde{h}^{[21]}_{21}\\ -\tilde{h}^{[11]*}_{21}&\tilde{h}^{[11]*}_{11}&-\tilde{h}^{[21]*}_{21}&\tilde{h}^{[21]*}_{11}\\
\tilde{h}^{[11]}_{12}&\tilde{h}^{[11]}_{22}&\tilde{h}^{[21]}_{12}&\tilde{h}^{[21]}_{22}\\
-\tilde{h}^{[11]*}_{22}&\tilde{h}^{[11]*}_{12}&-\tilde{h}^{[21]*}_{22}&\tilde{h}^{[21]*}_{12}\end{array}\right].
\end{align}
This is straightforward since the determinant of the above matrix is
a polynomial function of eight entries $\tilde{h}^{[j1]}_{mn}$ with
$j,m,n\in \{1,2\}$. Recall that $\tilde{h}^{[11]}_{mn}$ depends on
channel matrices $\mb{H}^{[11]}$ and $\mb{H}^{[12]}$, while
$\tilde{h}^{[21]}_{mn}$ depends on $\mb{H}^{[21]}$ and
$\mb{H}^{[22]}$, with all channel matrices being independently
drawn. The equivalent channels $\tilde{h}^{[11]}_{mn}$ are
independent from $\tilde{h}^{[21]}_{mn}$. Then, the determinant
polynomial is either 0 or non-zero for all values of
$\tilde{h}^{[j1]}_{mn}$ with probability 1\cite{Linear05}. When
$\tilde{h}^{[11]}_{11}=\tilde{h}^{[21]}_{12}=1$ and
$\tilde{h}^{[11]}_{21}=\tilde{h}^{[11]}_{12}=\tilde{h}^{[11]}_{22}=\tilde{h}^{[21]}_{11}=\tilde{h}^{[21]}_{21}=\tilde{h}^{[21]}_{22}=0$,
the matrix in \eqref{eq-matrix} becomes an identity matrix and full
rank. Thus, the determinant is not a zero polynomial and the matrix
in \eqref{eq-matrix} is full rank with probability 1.

For each data stream $s_k^{[ji]}$, a rate that grows linearly with
$\log P$ can be reliably supported. Since $8$ streams are sent over
the network in $3$ channel uses, the proposed scheme achieves the
DoF gain of $\frac{8}{3}$. The outerbound on the DoF gain of the
$2\times 2$ double-antenna X channel was characterized in
\cite{JaSh08} to be $\frac{8}{3}$. Therefore, the proposed scheme
achieves the maximum DoF gain.
\end{proof}

\section{Alamouti-coded Transmission for Cellular
Networks}\label{Sec-Cellular} In this section, we discuss two types
of cellular networks: the IMAC and IBC networks\cite{Suh08}, where
interference from a neighboring cell degrades in-cell communication.
Again, the use of Alamouti codes together with IA can bring the
maximum transmission rate and a diversity gain of $2$. We explain
the network models and show the maximum DoF gain in Subsection
\ref{subsec-model}. Since the X channel is a special case of the
IMAC, we briefly describe its transmission in Subsection
\ref{subsec-IMAC}. Transmission in the IBC is more challenging
compared to the IMAC because of the required designs of imperfect
alignment. Description for alignment in IBC is contained in
Subsection \ref{subsec-interBC}. Regarding channel information, each
MS requires only the knowledge of the interferring link connected to
itself, and each base station (BS) needs channel information within
its cell as well as the knowledge of its MSs' beamformers.

\subsection{The IMAC and IBC network models}\label{subsec-model}

Consider a two-cell IMAC as illustrated in the left side of
Fig.~\ref{fig-InterBC}. In each cell, one BS serves two MSs. All
nodes are equipped with two antennas. In the IMAC, we can use the
receiver's index for the cell index, since there is only one
receiver in each cell. In Cell $i$, transmitter $j$ has independent
symbols $s_k^{[ji]}$ to send to Receiver $i$, where $i,j \in
\{1,2\}$. The desired links are described by channel matrix
$\mb{H}^{[ji]}$, where $\mb{H}^{[ji]}\in \mathds{C}^{2\times 2}$.
Due to the simultaneous transmission, Cell 1 creates co-channel
interference to Cell 2, and similarly does Cell 2 to Cell 1. The
interferring link from Transmitter $j$ to Cell $i$ is described by
channel matrix $\mb{I}^{[ji]}$, where $\mb{I}^{[ji]}\in
\mathds{C}^{2\times 2}$. We assume that all entries in channel
matrices have i.~i.~d.~$\mc{CN}(0,1)$ distribution, and remain
constant during the transmission. The reciprocal channel of the IMAC
is an IBC, where the directions of communication are reversed.
Contrary to the IMAC, in the IBC, we can use the transmitter's index
for the cell index, since there is only one transmitter in each
cell. Transmitter $j$ sends independent symbols $s_{k}^{[ji]}$ to
Receiver $i$ in Cell $j$ through link $\mb{H}^{[ji]}$, and
simultaneously interferes User $i$ in the other Cell $\bar{j}$
through link $\mb{I}^{[\bar{j}i]}$. We can use similar notations as
that of the IMAC for the IBC with an exchange of the cell and user
indices. The IBC and adopted notations are shown in the right side
of Fig.~\ref{fig-InterBC}.

IA is considered for the IMAC in \cite{Suh08} and the IBC in
\cite{Suh11,Shin11}. These two channel models are introduced for
frequency selective channels in \cite{Suh08}, where the duality
between these two channels is also demonstrated. Transmission in a
two-cell IBC is studied in \cite{Shin11} with the number of BS
antennas larger than the number of receive antennas. Our paper
considers a MIMO setting where all nodes have equal number of
antennas. First, we show the outerbound on the DoF gains. In the
proof, we assume that $s_{k}^{[ji]}$ operates in the long-term
region and carries one DoF gain.

\begin{theorem}
For a two-cell IBC with two users in each cell and two antennas at
each node, let $d^{[ji]}$ be the DoF gain sent from Transmitter $j$
to Receiver $i$ in Cell $j$. The DoF gain region $\mc{D}^{\mr{IBC}}$
is
\begin{align}\label{eq-R1}
d^{[11]}+d^{[21]}+d^{[22]}\le 2, \\
\label{eq-R12} d^{[12]}+d^{[21]}+d^{[22]}\le 2,
\\ \label{eq-R2} d^{[21]}+d^{[11]}+d^{[12]}\le 2, \\ \label{eq-R22}
d^{[22]}+d^{[11]}+d^{[12]}\le 2.
\end{align}
\end{theorem}
\begin{proof}
The proof is similar to that of the outerbound on X
channels\cite{JaSh08}. Since the network is symmetric for each cell
and each receiver, we only show inequality \eqref{eq-R1} and the
other three inequalities hold by similar arguments. We argue that
the DoF gain region
$\underset{\mc{D}^{\mr{IBC}}}{\max}\left(d^{[11]}+d^{[21]}+d^{[22]}\right)$
can be outerbounded by those of two channels illustrated in
Fig.~\ref{fig-IBCOuterbound}. The first outerbound is a modified IBC
without Receiver 2 in Cell 1. BS1 sends messages only to Receiver 1.
Obviously, any reliable coding schemes in the IBC can be used
reliably in the modified IBC. Then, let $\mc{D}^{\mr{IBC}'}$ denote
the DoF gain regions of the modified IBC, we have
$\underset{\mc{D}^{\mr{IBC}}}{\max}\left(d^{[11]}+d^{[21]}+d^{[22]}\right)\le
\underset{\mc{D}^{\mr{IBC}'}}{\max}\left(d^{[11]}+d^{[21]}+d^{[22]}\right)$.
We can further outerbound the DoF region of the modified IBC using
the Z channel by allowing receivers in Cell 2 to cooperate (right
side of Fig.~\ref{fig-IBCOuterbound}). This is because any reliable
coding schemes for the modified IBC can be used in the Z channel by
adding interference at receivers in Cell 2 and decoding as if R3 and
R4 are distributed. Let the DoF gain region of the Z channel be
$\mc{D}^{\mr{Z}}$. From Corollary 1 in \cite{JaSh08}, we have
$\underset{\mc{D}^{\mr{Z}}}{\max}\left(d^{[11]}+d^{[21]}+d^{[22]}\right)\le
2$, since both BS2 and R1 have two antennas. It follows
$\underset{\mc{D}^{\mr{IBC}}}{\max}\left(d^{[11]}+d^{[21]}+d^{[22]}\right)\le
\underset{\mc{D}^{\mr{Z}}}{\max}\left(d^{[11]}+d^{[21]}+d^{[22]}\right)\le
2$.
\end{proof}
Due to the duality between the IMAC and the IBC, the same DoF gain
region holds for the IMAC. Combing \eqref{eq-R1}, \eqref{eq-R12},
\eqref{eq-R2}, \eqref{eq-R22} results in
$d^{[11]}+d^{[12]}+d^{[21]}+d^{[22]}\le \frac{8}{3}$.

\subsection{Transmission methods in the IMAC}\label{subsec-IMAC}

The maximum rate for the considered IMAC is $\frac{8}{3}$ symbols
per channel use. Noticing that the double-antenna $2\times 2$ X
channel is a special scenario of the two-cell IMAC when
$\mb{I}^{[ji]}=\mb{H}^{[j\bar{i}]}$. Then, it is straightforward to
use the method we have proposed for the X channel for the two-cell
IMAC. Specifically, two symbols $k\in\{1,2\}$, encoded in Alamouti
codes, are transmitted in three channel uses. Transmission in Cell
$1$ occurs in the first two time slots, while transmission in Cell
$2$ occurs in the last two time slots. For Transmitter $j$ in Cell
$i$, the normalized inversion of $\mb{I}^{[ji]}$ is used as the
alignment precoder. Then, four interferring symbols are aligned into
two dimensions. Since the X channel is a special case of the
considered IMAC, a diversity gain of 2 is achievable at the maximum
rate of $\frac{8}{3}$ symbols per channel use. Diversity analysis
for the IMAC using the proposed method is similar to Theorem
\ref{thm-div}.

\subsection{Transmission methods in the IBC}\label{subsec-interBC}
In what follows, we discuss the extension to the two-cell IBC. By
duality of reciprocal channels, the maximum rate of the two-cell IBC
is also $\frac{8}{3}$ symbols per channel use. Let the transmission
duration $T$ be three channel uses. To achieve the maximum rate,
each transmitter sends two symbols to each receiver. In total, $8$
symbols are transmitted over the network in three channel uses,
which amounts to the rate of $\frac{8}{3}$ symbols per channel use.
Since each receiver is equipped with two antennas and receives in
three time slots, a six-dimensional signal space is created. Each
receiver intends to decode two symbols and leaves the remaining
four-dimensional subspace for six interfering symbols (two symbols
are for the other receiver in the same cell, i.e., intra-cell
interference, and four symbols are from the other cell, i.e.,
inter-cell interference). Thus, we need an alignment design that
aligns six symbols in four-dimensional subspace. Such an imperfect
alignment design cannot be trivially extended from the proposed
method for X channels, where interference is completely aligned.

We use the method constructing a dual system from the original
system as proposed in \cite{Li12}. The methodology has been used to
design the dual Alamouti codes and the downlink IC method, where
receiver processing is totally blind of channel information. The
constructed scheme can bring to the dual system the same diversity
gain as in the original system. We use the transmission method in
the IMAC as the original system to derive its dual system. The
derivation is involved, and we directly present the transmission
method in the two-cell IBC. Note that a diversity gain of 2 is
achievable for the dual system, following the definition of dual
systems with ZF designs (Definition 1 and Proposition 1 in
\cite{Li12}).

The system diagram is shown in Fig.~\ref{fig-IBCStr0}. Let the
transmit block be $\mb{X}^{[j]}$, where $\mb{X}^{[j]}\in
\mathds{C}^{3\times 2}$. The receive block at Receiver $i$ in Cell
$j$ can be written as
\begin{align}
\mb{Y}^{[ji]}=\mb{X}^{[j]}\mb{H}^{[ji]}+\mb{X}^{[\bar{j}]}\mb{I}^{[ji]}+\mb{W}^{[ji]},
\end{align}
where $\mb{W}^{[ji]}\in \mathds{C}^{3\times 2}$ denotes the
AWGN matrix. Different from the IMAC, we use the inversion of
the interferring link as the receive beamforming matrix
\begin{align}\label{eq-recbf}
\underset{\tilde{\mb{Y}}^{[ji]}}{\underbrace{\mb{Y}^{[ji]}\left(\mb{I}^{[ji]}\right)^{-1}}}=\mb{X}^{[j]}\underset{\tilde{\mb{H}}^{[ji]}}{\underbrace{\mb{H}^{[ji]}\left(\mb{I}^{[ji]}\right)^{-1}}}+\mb{X}^{[\bar{j}]}+\underset{\tilde{\mb{W}}^{[ji]}}{\underbrace{\mb{W}^{[ji]}\left(\mb{I}^{[ji]}\right)^{-1}}}.
\end{align}
By such receive beamforming matrices, the equivalent interferring
links are identical at both receivers in one cell. This helps the
design of alignment precoder, as will be explained later. The
transmitter design is based on the equivalent channel matrix
$\tilde{\mb{H}}^{[ji]}\in\mathds{C}^{2\times 2}$. Each transmitter
collects two symbols $s_k^{[ji]}$ ($k\in \{1,2\}$), modulated by PSK
constellations, for each receiver in the cell. The symbols are
encoded using Alamouti codes followed by linear precoding as
\begin{align}\label{eq-precode}
\left[\begin{array}{cc}x_{11}^{[j]}&x_{12}^{[j]}\\x_{21}^{[j]}&
x_{22}^{[j]}\end{array}\right]=\underset{\mb{S}^{[j1]}}{\underbrace{\left[\begin{array}{cc}s_{1}^{[j1]}&s_{2}^{[j1]}\\-s_{2}^{[j1]*}&s_{1}^{[j1]*}\end{array}\right]}}\mb{P}^{[j1]}
+\underset{\mb{S}^{[j2]}}{\underbrace{\left[\begin{array}{cc}s_{1}^{[j2]}&s_{2}^{[j2]}\\-s_{2}^{[j2]*}&s_{1}^{[j2]*}\end{array}\right]}}\mb{P}^{[j2]},
\end{align}
where $\mb{P}^{[ji]} \in \mathds{C}^{2\times 2}$ are the precoding
matrix for Receiver $i$ in Cell $j$. We use the precoding matrices
from the downlink IC method\cite{Li12}. Let the $(m,n)$the entry of
$\tilde{\mb{H}}^{[ji]}$ be $\tilde{h}^{[ji]}_{mn}$. The matrix
$\mb{P}^{[ji]}$ is designed as
\begin{align}\label{eq-precoding}
\mb{P}^{[ji]}=\alpha^{[ji]}\left(
\frac{\hat{\mb{H}}^{[ji]*}_1\hat{\mb{H}}^{[j\bar{i}]}_1}{\left\|\tilde{\mb{h}}^{[j\bar{i}]}_1\right\|^2}-
\frac{\hat{\mb{H}}^{[ji]*}_2\hat{\mb{H}}^{[j\bar{i}]}_2}{\left\|\tilde{\mb{h}}^{[j\bar{i}]}_2\right\|^2}\right)
\left[\begin{array}{cc}\frac{\tilde{h}^{[j\bar{i}]*}_{11}}{\left\|\tilde{\mb{h}}^{[j\bar{i}]}_1\right\|^2}&-\frac{\tilde{h}^{[j\bar{i}]*}_{21}}{\left\|\tilde{\mb{h}}^{[j\bar{i}]}_2\right\|^2}\\
\frac{\tilde{h}^{[j\bar{i}]*}_{12}}{\left\|\tilde{\mb{h}}^{[j\bar{i}]}_1\right\|^2}&-\frac{\tilde{h}^{[j\bar{i}]*}_{22}}{\left\|\tilde{\mb{h}}^{[j\bar{i}]}_2\right\|^2}\end{array}\right],
\end{align}
where $\alpha^{[ji]}\in \mathds{R}$ denotes a power control
parameter for Receiver $i$ in Cell $j$,
$\tilde{\mb{h}}^{[ji]}_m$ denotes the $m$th row in
$\tilde{\mb{H}}^{[ji]}$, and
\begin{align}\label{eq-ala}
\hat{\mb{H}}_m^{[ji]}=\left[\begin{array}{cc}\tilde{h}^{[ji]}_{m1}&\tilde{h}^{[ji]}_{m2}\\-\tilde{h}^{[ji]*}_{m2}&\tilde{h}^{[ji]*}_{m1}\end{array}\right],\
m\in \{1,2\}.
\end{align}
For the details behind the derivation of the designs in
\eqref{eq-precoding}, the interested reader is referred to
\cite{Li12}. Here, we only explain how alignment is created. The
symbols $x_{tm}^{[ji]}$ in \eqref{eq-precode} are rearranged to
generate the transmit block $\mb{X}^{[j]}$,
\begin{align}\label{eq-alignment}
\mb{X}^{[1]}=\left[\begin{array}{cc}x^{[1]}_{11}&x^{[1]}_{12}\\x^{[1]}_{21}& x^{[1]}_{22}\\
x^{[1]*}_{22}&-x^{[1]*}_{21}
\end{array}\right],\
\mb{X}^{[2]}=\left[\begin{array}{cc}-x^{[2]*}_{22}&x^{[2]*}_{21}\\x^{[2]}_{21}& x^{[2]}_{22}\\
x^{[2]}_{11}&x^{[2]}_{12}
\end{array}\right].
\end{align}
The four entries in the left-side of \eqref{eq-precode} carry four
independent symbols. Recall that the vertical dimension of
$\mb{X}^{[j]}$ refers to the temporal dimension. From
\eqref{eq-alignment}, Transmitter $1$ sends four symbols in the
first two time slots. In time Slot $3$, redundant symbols are
transmitted to make the submatrix in time Slots $2$ and $3$ have the
\emph{swapped Alamouti structure}, i.e.,
\begin{align}\label{eq-swapal}
\left[\begin{array}{cc}a&b\\ b^*&-a^*\end{array}\right],
\end{align}
which can be obtained by swapping the columns of an Alamouti matrix.
Transmitter $2$ sends four symbols in the last two time slots. In
time Slot $1$, redundant symbols are transmitted to make the
submatrix in time Slots $1$ and $2$ also have the swapped Alamouti
structure. It will be shown that the swapped Alamouti structure
aligns the interference as well.

Let us further discuss receiver operations. Let the $(t,n)$th entry
of $\tilde{\mb{Y}}^{[ji]}$ in \eqref{eq-recbf} be
$\tilde{y}^{[ji]}_{tn}$. The receivers in Cell $1$ extract useful
symbols using signals received in the first two time slots as
\begin{align}\label{eq-result1}
\left[\begin{array}{cc}\hat{y}^{[1i]}_1&
\hat{y}^{[1i]}_2\end{array}\right]=\left[\begin{array}{cc}\tilde{y}^{[1i]}_{11}+\tilde{y}^{[1i]*}_{22}&
\tilde{y}^{[1i]}_{12}-\tilde{y}^{[1i]*}_{21}\end{array}\right],\
i\in \{1,2\}.
\end{align} In Cell $2$, receivers calculate
$\left[\begin{array}{cc}\hat{y}^{[2i]}_1&
\hat{y}^{[2i]}_2\end{array}\right]=\left[\begin{array}{cc}\tilde{y}^{[2i]}_{31}+\tilde{y}^{[2i]*}_{22}&
\tilde{y}^{[2i]}_{32}-\tilde{y}^{[2i]*}_{21}\end{array}\right],\
i\in \{1,2\} $ using signals received in the last two time slots.
Decoding of symbol $s_k^{[ji]}$ is performed by $\underset{s}{\max}
\hat{y}^{[ji]}_k s^*.$ The simple receiver operations are due to the
precoder designs in \eqref{eq-precode}. From the receiver operations
in \eqref{eq-recbf}, \eqref{eq-result1}, and the decoding, only the
knowledge of $\mb{I}^{[ji]}$ is required at Receiver $i$ in Cell
$j$. Transmitter operations are based on $\tilde{\mb{H}}^{[ji]}$.
Then, the knowledge of $\mb{H}^{[ji]}$ and $\mb{I}^{[ji]}$ for $i\in
\{1,2\}$ is required at Transmitter $j$.

\subsubsection{Alignment pattern}In what follows, we explain how the proposed method aligns six symbols in a four-dimensional subspace and
how Alamouti designs are used to protect desired symbols. First, we
introduce some intermediate variables to simplify notations. Note
that from \eqref{eq-ala} and \eqref{eq-precode}, both the matrices
$\hat{\mb{H}}_m^{[ji]}$ and $\mb{S}^{[ji]}$ have the Alamouti
structure. Since matrix multiplication and addition are closed for
two Alamouti matrices, we can define $c_{k}^{[ji]}\in \mathds{C}$
\begin{align}\label{eq-rotation}
\left[\begin{array}{cc}c_{1}^{[ji]}&c_{2}^{[ji]}\\-c_{2}^{[ji]*}&c_{1}^{[ji]*}\end{array}\right]=\alpha^{[ji]}\mb{S}^{[ji]}\left(
\frac{\hat{\mb{H}}^{[ji]*}_1\hat{\mb{H}}^{[j\bar{i}]}_1}{\left\|\tilde{\mb{h}}^{[j\bar{i}]}_1\right\|^2}-
\frac{\hat{\mb{H}}^{[ji]*}_2\hat{\mb{H}}^{[j\bar{i}]}_2}{\left\|\tilde{\mb{h}}^{[j\bar{i}]}_2\right\|^2}\right)
\end{align}
as the rotated symbols of $s_k^{[ji]}$. Without loss of generality,
we only show alignment at receivers in Cell $1$. Using $c_k^{[ji]}$,
we can expand the receive signals of Receiver 1 in \eqref{eq-recbf}
using $c_{k}^{[ji]}$ as
{\setlength{\arraycolsep}{3pt}\begin{align}\nonumber
&\left[\begin{array}{cc}\tilde{y}_{11}^{[11]}&\tilde{y}_{12}^{[11]}\\
\tilde{y}_{21}^{[11]}&\tilde{y}_{22}^{[11]}\end{array}\right]=\left(\sum_{i=1,2}\left[\begin{array}{cc}c^{[1i]}_1&c^{[1i]}_2\\-c^{[1i]*}_{2}&c^{[1i]*}_{1}\end{array}\right]\left[\begin{array}{cc}\frac{\tilde{h}^{[1\bar{i}]*}_{11}}{\left\|\tilde{\mb{h}}^{[1\bar{i}]}_1\right\|^2}&-\frac{\tilde{h}^{[1\bar{i}]*}_{21}}{\left\|\tilde{\mb{h}}^{[1\bar{i}]}_2\right\|^2}\\
\frac{\tilde{h}^{[1\bar{i}]*}_{12}}{\left\|\tilde{\mb{h}}^{[1\bar{i}]}_1\right\|^2}&-\frac{\tilde{h}^{[1\bar{i}]*}_{22}}{\left\|\tilde{\mb{h}}^{[1\bar{i}]}_2\right\|^2}\end{array}\right]\right)\left[\begin{array}{cc}\tilde{h}^{[11]}_{11}&\tilde{h}^{[11]}_{12}\\
\tilde{h}^{[11]}_{21}&\tilde{h}^{[11]}_{22}\end{array}\right]+\left[\begin{array}{cc}-x^{[2]*}_{22}&x^{[2]*}_{21}\\
x^{[2]}_{21}&x^{[2]}_{22}\end{array}\right]+\tilde{\mb{W}}^{[1i]}_1\\
\nonumber
&=\left[\begin{array}{cc}c^{[11]}_1&c^{[11]}_2\\-c^{[11]*}_{2}&c^{[11]*}_{1}\end{array}\right]
\underset{{\mc{H}}^{[11]}}{\underbrace{\left[\begin{array}{cc}\frac{\tilde{h}^{[12]*}_{11}\tilde{h}^{[11]}_{11}}{\left\|\tilde{\mb{h}}^{[12]}_1\right\|^2}-\frac{\tilde{h}^{[12]*}_{21}\tilde{h}^{[11]}_{21}}{\left\|\tilde{\mb{h}}^{[12]}_2\right\|^2}
&\frac{\tilde{h}^{[12]*}_{11}\tilde{h}^{[11]}_{12}}{\left\|\tilde{\mb{h}}^{[12]}_1\right\|^2}-\frac{\tilde{h}^{[12]*}_{21}\tilde{h}^{[11]}_{22}}{\left\|\tilde{\mb{h}}_2^{[12]}\right\|^2}\\
\frac{\tilde{h}^{[12]*}_{12}\tilde{h}^{[11]}_{11}}{\left\|\tilde{\mb{h}}^{[12]}_1\right\|^2}-\frac{\tilde{h}^{[12]*}_{22}\tilde{h}^{[11]}_{21}}{\left\|\tilde{\mb{h}}_2^{[12]}\right\|^2}
&\frac{\tilde{h}^{[12]*}_{12}\tilde{h}^{[11]}_{12}}{\left\|\tilde{\mb{h}}^{[12]}_1\right\|^2}-\frac{\tilde{h}^{[12]*}_{22}\tilde{h}^{[11]}_{22}}{\left\|\tilde{\mb{h}}_2^{[12]}\right\|^2}\end{array}\right]}}
\\
&+\left[\begin{array}{cc}c^{[12]}_1&c^{[12]}_2\\-c^{[12]*}_{2}&c^{[12]*}_{1}\end{array}\right]
\underset{\mc{H}^{[12]}}{\underbrace{\left[\begin{array}{cc}\frac{\left|\tilde{h}^{[11]}_{11}\right|^2}{\left\|\tilde{\mb{h}}_1^{[11]}\right\|^2}-\frac{\left|\tilde{h}^{[11]}_{21}\right|^2}{\left\|\tilde{\mb{h}}^{[11]}_2\right\|^2}
&\frac{\tilde{h}^{[11]*}_{11}\tilde{h}^{[11]}_{12}}{\left\|\tilde{\mb{h}}_1^{[11]}\right\|^2}-\frac{\tilde{h}^{[11]*}_{21}\tilde{h}^{[11]}_{22}}{\left\|\tilde{\mb{h}}^{[11]}_2\right\|^2}\\
\frac{\tilde{h}^{[11]*}_{12}\tilde{h}^{[11]}_{11}}{\left\|\tilde{\mb{h}}_1^{[11]}\right\|^2}-\frac{\tilde{h}^{[11]*}_{22}\tilde{h}^{[11]}_{21}}{\left\|\tilde{\mb{h}}^{[11]}_2\right\|^2}
&\frac{\left|\tilde{h}^{[11]}_{12}\right|^2}{\left\|\tilde{\mb{h}}_1^{[11]}\right\|^2}-\frac{\left|\tilde{h}^{[11]}_{22}\right|^2}{\left\|\tilde{\mb{h}}^{[11]}_2\right\|^2}\end{array}\right]}}
+\left[\begin{array}{cc}-x^{[2]*}_{22}&x^{[2]*}_{21}\\
x^{[2]}_{21}&x^{[2]}_{22}\end{array}\right]+\tilde{\mb{W}}^{[1i]}_1.\label{eq-time12}
\end{align}}
\hspace{-4pt}The matrices ${\mc{H}}^{[11]}\in \mathds{C}^{2\times
2}$ and ${\mc{H}}^{[12]}\in \mathds{C}^{2\times 2}$ are the
equivalent channel matrices for $c_k^{[11]}$ and $c_k^{[12]}$,
respectively. Let the $(m,n)$th entry of ${\mc{H}}^{[11]}$ and
${\mc{H}}^{[12]}$ be $\underline{h}^{[11]}_{mn}$ and
$\underline{h}^{[12]}_{mn}$, respectively. It can be verified that
{\setlength{\arraycolsep}{2pt}\small\begin{align*}
&\underline{h}^{[12]}_{11}+\underline{h}^{[12]*}_{22}=\frac{\left|\tilde{h}^{[11]}_{11}\right|^2}{\left\|\tilde{\mb{h}}_1^{[11]}\right\|^2}-\frac{\left|\tilde{h}^{[11]}_{21}\right|^2}{\left\|\tilde{\mb{h}}^{[11]}_2\right\|^2}+\left(\frac{\left|\tilde{h}^{[11]}_{12}\right|^2}{\left\|\tilde{\mb{h}}_1^{[11]}\right\|^2}-\frac{\left|\tilde{h}^{[11]}_{22}\right|^2}{\left\|\tilde{\mb{h}}^{[11]}_2\right\|^2}\right)^*=\frac{\left|\tilde{h}^{[11]}_{11}\right|^2+\left|\tilde{h}^{[11]}_{12}\right|^2}{\left\|\tilde{\mb{h}}_1^{[11]}\right\|^2}-\frac{\left|\tilde{h}^{[11]}_{21}\right|^2+\left|\tilde{h}^{[11]}_{22}\right|^2}{\left\|\tilde{\mb{h}}^{[11]}_2\right\|^2}=0
\\
&\underline{h}^{[12]}_{12}-\underline{h}^{[12]*}_{21}=\frac{\tilde{h}^{[11]*}_{11}\tilde{h}^{[11]}_{12}}{\left\|\tilde{\mb{h}}_1^{[11]}\right\|^2}-\frac{\tilde{h}^{[11]*}_{21}\tilde{h}^{[11]}_{22}}{\left\|\tilde{\mb{h}}^{[11]}_2\right\|^2}-\left(\frac{\tilde{h}^{[11]*}_{12}\tilde{h}^{[11]}_{11}}{\left\|\tilde{\mb{h}}_1^{[11]}\right\|^2}-\frac{\tilde{h}^{[11]*}_{22}\tilde{h}^{[11]}_{21}}{\left\|\tilde{\mb{h}}^{[11]}_2\right\|^2}
\right)^*=0.
\end{align*}}
\hspace{-3pt}Thus, the matrix ${\mc{H}}^{[12]}$ has the swapped
Alamouti structure that has been defined in \eqref{eq-swapal}. Now,
let us explain the use of the swapped Alamouti structure to pad the
transmit block in \eqref{eq-alignment}. From \eqref{eq-time12}, all
interferring symbols are carried in $c^{[12]}_1, c^{[12]}_2,
x_{11}^{[2]},x_{12}^{[2]},x^{[2]}_{21}$, and $x^{[2]}_{22}$. Note
that in \eqref{eq-time12}, the rotated symbols $c^{[12]}_1$ and
$c^{[12]}_2$ have the Alamouti structure. It can be verified that
multiplying the Alamouti matrix containing $c^{[12]}_1$ and
$c^{[12]}_2$ with the matrix ${\mc{H}}^{[12]}$ still has the swapped
Alamouti structure. Also from \eqref{eq-time12}, the interfering
symbols from Cell 2, i.e., $x^{[2]}_{21},x^{[2]}_{22}$, are placed
in a matrix having the swapped Alamouti structure (it is created by
the padding in \eqref{eq-alignment}). Therefore, all six interfering
symbols are aligned on the swapped Alamouti structure, which
occupies only a two-dimensional subspace in the four-dimensional
signal space (we only consider two receive time slots). Adding
receive signals in time Slot 3 at most expands the dimension of the
interference subspace from two to four. Then, we are able to align
six interferring symbols in a four-dimensional subspace. This
intuitively explains the alignment pattern. To see a complete
picture of the receive signal space, we can expand the receive
signals in time Slot $3$ in \eqref{eq-recbf} using $c_k^{[ji]}$ as
{\setlength{\arraycolsep}{3pt}\begin{align}\nonumber
&\left[\begin{array}{cc}\tilde{y}_{31}^{[11]}&\tilde{y}_{32}^{[11]}\end{array}\right]=
\left[\begin{array}{cc}x^{[1]*}_{22}&-x^{[1]*}_{21}\end{array}\right]\left[\begin{array}{cc}\tilde{h}^{[11]}_{11}&\tilde{h}^{[11]}_{12}\\
\tilde{h}^{[11]}_{21}&\tilde{h}^{[11]}_{22}\end{array}\right]+\left[\begin{array}{cc}
x^{2}_{11}&x^2_{12}\end{array}\right]+\tilde{\mb{w}}_2^{[11]}\\\nonumber
&=\left[\begin{array}{cc}c^{[11]}_{2}&-c^{[11]}_{1}\end{array}\right]\underset{\overline{\mb{H}}^{[11]}}{\underbrace{\left[\begin{array}{cc}\frac{\tilde{h}^{[12]}_{21}}{\left\|\tilde{\mb{h}}^{[12]}_2\right\|^2}&\frac{\tilde{h}^{[12]}_{11}}{\left\|\tilde{\mb{h}}^{[12]}_1\right\|^2}\\
\frac{\tilde{h}^{[12]}_{22}}{\left\|\tilde{\mb{h}}^{[12]}_2\right\|^2}&\frac{\tilde{h}^{[12]}_{12}}{\left\|\tilde{\mb{h}}^{[12]}_1\right\|^2}\end{array}\right]
\left[\begin{array}{cc}\tilde{h}^{[11]}_{11}&\tilde{h}^{[11]}_{12}\\
\tilde{h}^{[11]}_{21}&\tilde{h}^{[11]}_{22}\end{array}\right]}}+\left[\begin{array}{cc}c^{[12]}_{2}&-c^{[12]}_{1}\end{array}\right]\underset{\overline{\mb{H}}^{[12]}}{\underbrace{\left[\begin{array}{cc}\frac{\tilde{h}^{[11]}_{21}}{\left\|\tilde{\mb{h}}^{[11]}_2\right\|^2}&\frac{\tilde{h}^{[11]}_{11}}{\left\|\tilde{\mb{h}}^{[11]}_1\right\|^2}\\
\frac{\tilde{h}^{[11]}_{22}}{\left\|\tilde{\mb{h}}^{[11]}_2\right\|^2}&\frac{\tilde{h}^{[11]}_{12}}{\left\|\tilde{\mb{h}}^{[11]}_1\right\|^2}\end{array}\right]
\left[\begin{array}{cc}\tilde{h}^{[11]}_{11}&\tilde{h}^{[11]}_{12}\\
\tilde{h}^{[11]}_{21}&\tilde{h}^{[11]}_{22}\end{array}\right]}}\\
&\quad +\left[\begin{array}{cc}
x^{[2]}_{11}&x^{[2]}_{12}\end{array}\right]+\tilde{\mb{w}}_2^{[11]}\label{eq-time3}
\end{align}}
\hspace{-3pt}Denote the $(m,n)$th entry of
$\overline{\mb{H}}^{[11]}$ and $\overline{\mb{H}}^{[12]}$ as
$\overline{h}^{[11]}_{mn}$ and $\overline{h}^{[12]}_{mn}$,
respectively. Combining \eqref{eq-time12} and \eqref{eq-time3}, we
can obtain the equivalent vector system equation at Receiver $1$ as
{\setlength{\arraycolsep}{1pt}\begin{align}\nonumber
&\left[\begin{array}{c}\tilde{y}^{[11]}_{11}\\\tilde{y}^{[11]}_{12}\\\tilde{y}^{[11]*}_{21}\\\tilde{y}^{[11]*}_{22}\\\tilde{y}^{[11]}_{31}\\\tilde{y}^{[11]}_{32}\end{array}\right]=
\left[\begin{array}{cc}\underline{h}_{11}^{[11]}&\underline{h}^{[11]}_{21}\\
\underline{h}^{[11]}_{12}&\underline{h}^{[11]}_{22}\\
\underline{h}^{[11]*}_{21}&-\underline{h}^{[11]*}_{11}\\
\underline{h}^{[11]*}_{22}&-\underline{h}^{[11]*}_{12}\\-\overline{h}_{21}^{[11]}&\overline{h}_{11}^{[11]}\\-\overline{h}_{22}^{[11]}&\overline{h}_{12}^{[11]}\end{array}\right]\left[\begin{array}{c}c^{[11]}_1\\c^{[11]}_2\end{array}\right]+\left[\begin{array}{cc}
\underline{h}^{[12]}_{11}&\underline{h}^{[12]}_{21}\\
\underline{h}^{[12]*}_{21}&-\underline{h}^{[12]*}_{11}\\
\underline{h}^{[12]*}_{21}&-\underline{h}^{[12]*}_{11}\\
-\underline{h}^{[12]}_{11}&-\underline{h}^{[12]}_{21}\\-\overline{h}_{[12]}^{21}&\overline{h}_{11}^{[12]}\\-\overline{h}_{22}^{[12]}&\overline{h}_{11}^{[12]}\end{array}\right]\left[\begin{array}{c}c^{[12]}_1\\c^{[12]}_2\end{array}\right]+
\left[\begin{array}{cccc}0&-1&0&0\\1&0&0&0\\1&0&0&0\\0&1&0&0\\0&0&1&0\\0&0&0&1\end{array}\right]\left[\begin{array}{c}x^{[2]*}_{21}\\x^{[2]*}_{22}\\x_{11}^{[2]}\\x_{12}^{[2]}\end{array}\right]+\left[\begin{array}{c}\tilde{w}^{[11]}_{11}\\\tilde{w}^{[11]}_{12}\\\tilde{w}^{[11]*}_{21}\\\tilde{w}^{[11]*}_{22}\\\tilde{w}^{[11]}_{31}\\\tilde{w}^{[11]}_{32}\end{array}\right]\\
&=\left[\begin{array}{cc}\underline{h}_{11}^{[11]}&\underline{h}^{[11]}_{21}\\
\underline{h}^{[11]}_{12}&\underline{h}^{[11]}_{22}\\
\underline{h}^{[11]*}_{21}&-\underline{h}^{[11]*}_{11}\\
\underline{h}^{[11]*}_{22}&-\underline{h}^{[11]*}_{12}\\-\overline{h}_{21}^{[11]}&\overline{h}_{11}^{[11]}\\-\overline{h}_{22}^{[11]}&\overline{h}_{12}^{[11]}\end{array}\right]\left[\begin{array}{c}c^{[11]}_1\\c^{[11]}_2\end{array}\right]+\underset{\mb{Q}}{\underbrace{\left[\begin{array}{cccc}0&-1&0&0\\1&0&0&0\\1&0&0&0\\0&1&0&0\\0&0&1&0\\0&0&0&1\end{array}\right]}}\left(\left[\begin{array}{cc}
\underline{h}^{[12]*}_{21}&-\underline{h}^{[12]*}_{11}\\
-\underline{h}^{[12]}_{11}&-\underline{h}^{[12]}_{21}\\-\overline{h}_{[12]}^{21}&\overline{h}_{11}^{[12]}\\-\overline{h}_{22}^{[12]}&\overline{h}_{11}^{[12]}\end{array}\right]\left[\begin{array}{c}c^{[12]}_1\\c^{[12]}_2\end{array}\right]+
\left[\begin{array}{c}x^{[2]*}_{21}\\x^{[2]*}_{22}\\x_{11}^{[2]}\\x_{12}^{[2]}\end{array}\right]\right)+\left[\begin{array}{c}\tilde{w}^{[11]}_{11}\\\tilde{w}^{[11]}_{12}\\\tilde{w}^{[11]*}_{21}\\\tilde{w}^{[11]*}_{22}\\\tilde{w}^{[11]}_{31}\\\tilde{w}^{[11]}_{32}\end{array}\right].\label{eq-sys0}
\end{align}}
\hspace{-2pt}From \eqref{eq-sys0}, the interfering symbols to
Receiver $1$ $\left(c^{[12]}_1, c^{[12]}_2,
x_{11}^{[2]},x_{12}^{[2]},x^{[2]}_{21},x^{[2]}_{22}\right)$ are
aligned in a four-dimensional subspace spanned by the columns of
$\mb{Q}$. The two desired symbols $c^{[11]}_1$ and $c^{[11]}_2$ are
located in the remaining two-dimensional subspace. The alignment
pattern is illustrated in Fig.~\ref{fig-align}. To cancel the
aligned interference, the receiver discards $\tilde{y}^{[11]}_{31}$
and $\tilde{y}^{[11]}_{32}$, then conducts the calculation in
\eqref{eq-result1}, i.e.,
\begin{align}\label{eq-ZF2}
\left[\begin{array}{c}\tilde{y}^{[11]}_{11}+\tilde{y}^{[11]*}_{22}\\
\tilde{y}^{[11]}_{12}-\tilde{y}^{[11]*}_{21}\end{array}\right]=\left[\begin{array}{cc}\underline{h}_{11}^{[11]}+\underline{h}^{[11]*}_{22}&\underline{h}_{21}^{[11]}-\underline{h}^{[11]*}_{12}\\
\underline{h}^{[11]}_{12}-\underline{h}^{[11]*}_{21}&\underline{h}^{[11]}_{22}+\underline{h}^{[11]*}_{11}\end{array}\right]\left[\begin{array}{c}c^{[11]}_1\\c^{[11]}_2\end{array}\right]+\left[\begin{array}{c}\tilde{{w}}_{11}^{[11]}+\tilde{{w}}_{22}^{[11]*}\\
\tilde{{w}}_{12}^{[11]}-\tilde{{w}}_{21}^{[11]*}\end{array}\right].
\end{align}
Two desired symbols occupy only a two-dimensional subspace with the
equivalent channel matrix having the Alamouti structure\footnote{In
addition, the rotation in \eqref{eq-rotation} diagonalizes the
equivalent channel matrix in \eqref{eq-ZF2}, thus resulting in
symbol-by-symbol decoding. For more details, the interested reader
is referred to Proposition 2 in \cite{Li12}.}. Then, the desired
symbols are protected by orthogonal channel vectors due to the
Alamouti design.

To summarize the key elements of alignment at Receiver $1$ in Cell
$1$, the precoding matrix used in \eqref{eq-precoding} creates an
equivalent channel matrix ${\mc{H}}^{[12]}$ with the swapped
Alamouti structure for the interferring symbol $s^{[12]}_k$.
Transmitter $2$ aligns to this structure by padding the transmit
block $\mb{X}^{[2]}$ in time Slot $1$. By using the inversion of the
interfering link, six interfering symbols $s^{[12]}_k$,
$s^{[21]}_k$, and $s^{[22]}_k$ are able to align in a
four-dimensional subspace at both receivers in one cell. It can be
verified that such alignment also occurs in Cell $2$. Specifically,
alignment is created by the swapped Alamouti structure in time Slots
$2$ and $3$ of $\mb{X}^{[1]}$.

%\subsection{The BC with common message}\label{subsec-compBC}
%
%\begin{figure}
%\centering
%  % Requires \usepackage{graphicx}
%  \includegraphics[width=5in]{images/CompBC}\\
%  \caption{Network models of the $2$-cell compound BC. One BS transmits independent symbols to two cells. The two users in one cell desire the same symbols.}\label{fig-CompBC}
%\end{figure}
%This subsection focus on the multi-antenna BC with common
%messages. The network model is shown in Fig.~\ref{fig-CompBC}.
%The BS transmits signals to two cells and each cell has two
%receivers. The receivers in the same cell desire the same
%messages, i.e., common message, whereas the receivers in
%different cells desire independent messages. For the
%considered multi-antenna settings, the BS is equipped with
%four antennas, and each receiver is equipped with two
%antennas. We use $j$ as the cell index and $i$ as the receiver
%index in each cell. The message to Cell $j$ is carried in
%symbol $s_k^{[j]}$, modulated by normalized constellation. The
%channels to Receiver $i$ in Cell $j$ is denoted as
%$\mb{H}^{[ji]}\in \mathds{C}^{4\times 2}$. The BC with common
%messages is also literally known as the compound BC, which is
%used to model CSIT uncertainty\cite{}. The obtained results
%directly apply to the compound BC. Here, we focus on the
%description for the BC with common messages, as it is widely
%applicable in cellular networks.
%
%{\color{red}{Known results for this channel model}}

\section{Simulation results}\label{Sec-Simulation}
In this section, we compare the proposed methods with related
transmission schemes in both the short-term regime and the long-term
regime. Throughout this section, the horizontal axis in all figures
represents SNR measured in dB. Since the noises are normalized and
the transmit power of each user is $P$, the SNR of the network is
$P$.

Simulations in the short-term regime are performed for two network
models. We simulate the average BER performance of the proposed
methods. Since the diversity gain is not changed by using any
channel codes, we simulate an uncoded system for simplicity. The
vertical axis represents the average BER. It is averaged over all
communication directions. The first group of simulations shows the
BER performance of the proposed alignment method using Alamouti
designs in X channels. For comparison, the JaSh scheme\cite{JaSh08}
is included. In addition, we have a new modified JaSh scheme that
has potential for diversity improvement. The modified JaSh scheme
uses Alamouti codes on top of the JaSh scheme. Recall that the JaSh
scheme creates a $2\times 2$ point-to-point channel after removing
the aligned interference and decoupling the symbols from the other
transmitter. The modified JaSh scheme uses an Alamouti code for the
$2\times 2$ channel to improve diversity while providing only half
of the symbol rate of the JaSh scheme. Uncoded symbols $s^{[ji]}_k$
are independently generated from a finite constellation. To achieve
the same bit rate, different modulations are used for the three
methods in Fig.~\ref{fig-XchanNew}. We use BPSK, BPSK, and QPSK
modulations for the proposed scheme, the JaSh scheme, and the
modified JaSh scheme, respectively, to achieve $2/3$ bits per
channel use per pair node (solid curves in Fig.~\ref{fig-XchanNew}).
Also, to include comparison at another bit rate, QPSK, QPSK, and
16PSK modulations are used for the proposed scheme, the JaSh scheme,
and the modified JaSh scheme, respectively, to achieve $4/3$ bits
per channel use per pair node (dashed curves in
Fig.~\ref{fig-XchanNew}).

Our proposed method achieves a diversity gain of 2, whereas the JaSh
scheme achieves a diversity gain of 1. These results verify the
analysis in Subsection \ref{subsec-analysis}. It can be observed
that the diversity benefits bring more than $10$~dB gain at
BER=$10^{-3}$ for both transmission rates. The modified JaSh scheme
cannot bring diversity improvement: only a diversity gain of 1 is
observed from Fig.~\ref{fig-XchanNew}. This is because the $2\times
2$ diagonal channel after removing aligned interference and
decoupling symbols has correlated diagonal entries. The sum of the
achievable SNRs on each channel is upperbounded by a term providing
a diversity of only 1. The proof for the diversity gain of the
modified JaSh scheme is provided in Appendix \ref{appen2}.
Consequently, simply using Alamouti codes on top of the JaSh scheme
cannot bring diversity improvement.

The second group of simulations compares the extended scheme with
the downlink IA \cite{Suh11} in the two-cell IBC. Note that in our
setting, each node has two antennas and two symbols are transmitted
to each receiver; while in \cite{Suh11}, each node has one antenna
and one symbol is transmitted to each receiver. We extend the
downlink IA method in \cite{Suh11} to our double-antenna setting to
achieve the same symbol rate. The system diagram is shown in
Fig.~\ref{fig-IBCStr}. The BS uses two transmit precoders: a random
precoder $\mb{P}$ and a ZF precoder $\mb{B}^{[j]}$ to null out
intra-cell interference. Each receiver utilizes a receive beamformer
$\mb{u}^{[ji]}$ to zero-force inter-cell interference. Specifically,
the BS sends two symbols to each receiver in three symbol
extensions, which creates a six-dimensional signal space. The
receive beamformer $\mb{u}^{[ji]}\in \mathds{C}^{2\times 6}$ rejects
four interferring symbols from the other cell by zero-forcing the
equivalent channel matrix $(\mb{I}_3\otimes \mb{I}^{[ji]})\mb{P}$
and accepts two desired symbols. The entries in the random precoder
$\mb{P}\in \mathds{C}^{6\times 4}$ are assumed
i.~i.~d.~$\mc{CN}(0,1)$ distributed. The ZF precoder
$\mb{B}\in\mathds{C}^{4\times 4}$ cancels the intra-cell
interference by ZF precoding over the equivalent channels
$\mb{u}^{[ji]}(\mb{I}_3\otimes \mb{I}^{[ji]})\mb{P}$. For channel
information requirements, both alignment methods need the knowledge
of the interferring link at the receivers, and the transmitters
require channel information within each cell in addition to the
knowledge of the receive beamformers. Since both alignment methods
have the same symbol rate, BPSK is used to achieve $2/3$ bits per
channel use per receiver (solid curves in Fig.~\ref{fig-IBCNew}),
and QPSK is used to achieve $4/3$ bits per channel use per receiver
(dashed curves in Fig.~\ref{fig-IBCNew}).

Fig.~\ref{fig-IBCNew} exhibits the comparison. Our proposed method
can achieve a diversity gain of 2, which provides an approximate
array gain of $20$~dB at $\mr{BER}=10^{-2}$, compared to the
downlink IA method.

In the long-term regime, we simulate and compare the achievable
ergodic mutual information for the related methods. An
i.~i.~d.~Gaussian codebook is used for each symbol $s_k^{[ji]}$. The
vertical axis represents the sum rate (measured in bits per channel
use) over all communication directions. Figs.~\ref{fig-XCap} and
\ref{fig-IBCCap} show the ergodic mutual information for the X
channel and the IBC, respectively. We can first observe that the
proposed method achieves the same DoF gain as the JaSh scheme in
Fig.~\ref{fig-XCap}, and as the downlink IA method in
Fig.~\ref{fig-IBCCap}. Additionally, in the entire SNR regime, our
proposed method has a better SNR offset compared to the previous
methods. For example, in Fig.~\ref{fig-XCap}, the proposed method
outperforms the JaSh scheme by approximately $3$ bits/channel use at
$\mr{SNR}=25$ dB; in Fig.~\ref{fig-IBCCap}, the proposed method
enjoys approximately $8$ bits/channel use gain over the downlink IA
method at $\mr{SNR}=25$ dB. Similar gains are also achieved in the
low SNR range. For all compared methods, a ZF receiver is used to
cancel the aligned interference as well as decouple the desired
signals. Since our proposed method incorporates orthogonal designs
between the two symbols from the same user, a ZF receiver does not
incur SNR loss when separating these two symbols. On the other hand,
for the previous proposed methods, such an SNR loss occurs during
the symbol separation. This intuitively explains the SNR gain in the
entire SNR range.

\section{Conclusions}\label{Sec-Conclusion}
In this paper, we have proposed a transmission scheme that achieves
the maximum symbol-rate, i.e., $\frac{2}{3}$ from node-to-node, with
high reliability for the double-antenna $2\times 2$ X channel. The
alignment scheme incorporates Alamouti designs before using the
normalized inversion of the cross channel as the transmit beamformer
to align symbols at unintended receivers. Each receiver removes
aligned interference followed by symbol decoupling using IC.
Consequently, a symbol-by-symbol decoding complexity is achieved at
both receivers. Both simulation and analysis demonstrate a diversity
gain of 2 for the symbol-by-symbol decoding in the proposed scheme.
This implies that a diversity gain of higher than 1 is achievable in
the short-term regime, yet simultaneously with the maximum DoF gain
in the long-term regime. The proposed transmission scheme has also
been extended to two cellular networks, the IMAC and IBC, to bring
the maximum-rate transmission with a diversity gain of 2.
Significant BER performance improvement is observed through
simulation compared to the downlink IA method. Further extension to
the two-user X channels with more than 2 antennas at each node is
also doable by sending multiple groups of Alamouti codes for each
communication direction.

We have also identified that designing alignment for diversity is
not straightforward. Using STBCs on top of the previous alignment
method in \cite{JaSh08} can neither bring diversity improvements nor
maintain the maximum DoF gain for the two-user X channel. This calls
for an optimization of existing alignment methods to jointly
consider the DoF gain and the diversity gain.

Note that the considered network has 2 antennas at each node. The
achievable diversity is upperbounded by the corresponding
point-to-point channel. In other words, the maximum diversity gain
for the considered network is $2\times 2=4$. Our proposed scheme
only achieves the full transmit diversity, whereas the receive
diversity gain is only 1. We do not claim that the proposed scheme
is optimal in terms of the diversity gain. Since our proposed scheme
separates the desired symbols by ZF, it is possible to further
improve the receive diversity by a joint-decoding of 4 desired
symbols at each receiver. Since the proposed method needs only four
symbol-by-symbol decodings, the expense of the joint-decoding
algorithm is the increased decoding complexity. We conjecture that
such a joint decoding will result in a diversity gain of 4.

To embed Alamouti codes into alignment, the network is required to
have infinitely many alignment modes, because Alamouti codes are
rotationally invariant. The discussed network models have redundant
transmit dimensions. Our design uses a normalized inversion of the
cross channels (See Eq.\eqref{eq-beamformer}) to constrain the
interference subspace to be an identity matrix. In general, the
interference subspace can be arbitrarily chosen, thus generating
infinitely many alignment modes. Unfortunately, some interference
networks, e.g., the interference channels without symbol extensions,
have finitely many alignment modes at the maximum DoF gain. Thus, it
is not clear how to improve their diversity gains by utilizing
orthogonal designs.

\section*{Acknowledge}
The authors would like to thank Syed A. Jafar for insightful
discussions on interference alignment schemes.

\renewcommand{\baselinestretch}{1}
\bibliographystyle{ieeetran}
\bibliography{IEEEabrv,IA,IC-Relay-TDMA}

\renewcommand{\baselinestretch}{1.8}
\useRomanappendicesfalse
\appendices
\section{Two useful lemmas}
To prove Theorem \ref{thm-prediv}, we need some lemmas.
\begin{lemma}\label{lemma1}
Let the entries of $\mb{F}\in \mathds{C}^{2\times 2}$ be i.~i.~d.
$\mc{CN}(0,1)$ distributed. The following instantaneous normalized
receive SNR
\begin{align}\label{eq-SNR}
\gamma=\frac{1}{\tr\left(\mb{F}^{-1}(\mb{F}^{-1})^*\right)}
\end{align}
provides diversity gain 1.
\end{lemma}
\begin{proof}
Let the singular values of $\mb{F}$ be $\lambda_1$ and $\lambda_2$
such that $\lambda_1\ge \lambda_2$. Eqn.~\eqref{eq-SNR} can be
expanded as
\begin{align*}
\gamma=\frac{1}{\tr\left(\mb{F}^{-1}(\mb{F}^{-1})^*\right)}=\frac{1}{\frac{1}{\lambda_1^2}+\frac{1}{\lambda_2^2}}<\frac{1}{\frac{1}{\lambda_2^2}}=\lambda_2^2.
\end{align*}
Since the smaller singular value $\lambda_2$ carries diversity 1
only\cite{DivSVD}, the diversity gain of $\gamma$ is upperbounded by
1. Further, we can lowerbound $\gamma$ as
\begin{align*}
\gamma=\frac{1}{\frac{1}{\lambda_1^2}+\frac{1}{\lambda_2^2}}>\frac{1}{\frac{1}{\lambda_2^2}+\frac{1}{\lambda_2^2}}=\frac{\lambda_2^2}{2}.
\end{align*}
Thus, the instantaneous normalized receive SNR in \eqref{eq-SNR} is
lowerbounded by a term with diversity 1. Therefore, the achievable
diversity for $\gamma$ is exactly 1.
\end{proof}

\begin{lemma}\label{lemma2}
Consider the following $N\times 1$ vector system equation
\begin{align}
\mb{y}=\mb{h}_1s_1+\sum_{i=2:M}\mb{h}_is_i+\mb{w},
\end{align}
where $\mb{w}\in\mathds{C}^{N\times 1}$ have i.~i.~d.~$\mc{CN}(0,1)$
distributed entries, and the channel vectors
$\mb{h}_i\in\mathds{C}^{N\times 1}$ are linearly independent. Let
$\mb{Q}\in \mathds{C}^{(N-M+1)\times N}$ be any full-rank ZF matrix
such that
\begin{align}\label{eq-ZF}
\mb{Q}\mb{h}_i=\mb{0}_{N-M+1}, \ i\in\{2,\ldots,M\}.
\end{align}
After ZF, the equivalent channel vector is $\mb{Q}\mb{h}_1$ and the
noise covariance matrix is $\mb{Q}\mb{Q}^*$. For any designs of
$\mb{Q}$, the resulting instantaneous normalized receive SNR after
ZF is
\begin{align}
\gamma=\mb{h}_1^*\mb{Q}^*\left(\mb{Q}\mb{Q}^*\right)^{-1}\mb{Q}\mb{h}_1=\mb{h}_1^*\mb{\Sigma}\mb{h}_1,
\end{align}
where $\mb{\Sigma}$ is the projection matrix to the null space of
$\left[\mb{h}_2,\ldots,\mb{h}_M\right]$. The instantaneous
normalized receive SNR is independent of the designs of $\mb{Q}$.
\end{lemma}

\begin{proof}
Let the SVD of $\mb{Q}$ be $\mb{Q}=\mb{U}\mb{\Lambda}\mb{V}$ where
$\mb{U}\in \mathds{C}^{(N-M+1)\times (N-M+1)}, \mb{V}\in
\mathds{C}^{N\times N}$ denote the singular vector matrix and
$\mb{\Lambda}\in \mathds{C}^{(N-M+1)\times N}$ denotes the singular
value matrix. Further denote
$\mb{\Lambda}=\left[\tilde{\mb{\Lambda}}\ \mb{0}\right]$ where
$\tilde{\mb{\Lambda}}\in \mathds{R}^{(N-M+1)\times (N-M+1)}$ denotes
the diagonal square matrix with all singular values. It follows
\begin{align*}
\mb{Q}^*\left(\mb{Q}\mb{Q}^*\right)^{-1}\mb{Q}&=
(\mb{U}\mb{\Lambda}\mb{V})^*\left((\mb{U}\mb{\Lambda}\mb{V})(\mb{U}\mb{\Lambda}\mb{V})^*\right)^{-1}(\mb{U}\mb{\Lambda}\mb{V})\\
&=\mb{V}^*\mb{\Lambda}^*\mb{U}^*\left(\mb{U}\mb{\Lambda}\mb{V}\mb{V}^*\mb{\Lambda}^*\mb{U}^*\right)^{-1}\mb{U}\mb{\Lambda}\mb{V}\\
&=\mb{V}^*\mb{\Lambda}^*\mb{U}^*\left(\mb{U}\mb{\Lambda}\mb{\Lambda}^*\mb{U}^*\right)^{-1}\mb{U}\mb{\Lambda}\mb{V}\\
&=\mb{V}^*\mb{\Lambda}^*\mb{U}^*\mb{U}\left(\mb{\Lambda}\mb{\Lambda}^*\right)^{-1}\mb{U}^*\mb{U}\mb{\Lambda}\mb{V}\\
&=\mb{V}^*\mb{\Lambda}^*\left(\mb{\Lambda}\mb{\Lambda}^*\right)^{-1}\mb{\Lambda}\mb{V}=\tilde{\mb{V}}^*\tilde{\mb{\Lambda}}^*\left(\tilde{\mb{\Lambda}}\tilde{\mb{\Lambda}}^*\right)^{-1}\tilde{\mb{\Lambda}}\tilde{\mb{V}}=\tilde{\mb{V}}^*\tilde{\mb{V}},
\end{align*}
where $\tilde{\mb{V}}$ denotes the first $N-M+1$ rows of $\mb{V}$.
It suffices to verify that $\tilde{\mb{V}}^*\tilde{\mb{V}}$ is the
projection matrix to the null space of the subspace spanned by
$\left[\mb{h}_2,\ldots,\mb{h}_M\right]$. For any vector
$\hat{\mb{h}}\in \mathds{C}^{N\times 1}$ located in the subspace of
$\left[\mb{h}_2,\ldots,\mb{h}_M\right]$, we can assume it to be
$\hat{\mb{h}}=\underset{i=2:M}{\sum}\mb{h}_ic_i$, where $c_i\in
\mathds{C}$ is an arbitrary coefficient. From the ZF constraint in
\eqref{eq-ZF}, we have
$\mb{Q}\mb{h}_i=\mb{U}\tilde{\mb{\Lambda}}\tilde{\mb{V}}\mb{h}_i=\mb{0}$.
Since $\mb{U}$ and $\tilde{\mb{\Lambda}}$ are invertible, it follows
that $\tilde{\mb{V}}\mb{h}_i=\mb{0}$. Thus, we have
\begin{align*}
\tilde{\mb{V}}^*\tilde{\mb{V}}\hat{\mb{h}}=\tilde{\mb{V}}^*\tilde{\mb{V}}\underset{i=2:M}{\sum}\mb{h}_ic_i=\underset{i=2:M}{\sum}\tilde{\mb{V}}^*\left(\tilde{\mb{V}}\right)\mb{h}_ic_i=\mb{0}.
\end{align*}
Note that the rows of $\tilde{\mb{V}}$ also form an orthonormal
basis for the considered null space. Therefore,
$\tilde{\mb{V}}^*\tilde{\mb{V}}$ is a projection matrix to the null
space of the subspace spanned by
$\left[\mb{h}_2,\ldots,\mb{h}_M\right]$.
\end{proof}
Lemma \ref{lemma2} says all ZF receivers are essentially the same in
terms of the output SNR. Therefore, to obtain general results for
any ZF receivers, we can rely on a special ZF receiver that
simplifies the analysis.

\section{Proof of Theorem \ref{thm-prediv}}\label{appenpre}
The proof is based on the outage probability of the instantaneous
normalized receive SNR $\gamma$ that has been defined in
\eqref{eq-div}. Since the network is statistically symmetric to each
symbol, without loss of generality, we only study the expression of
$\gamma$ for $s^{[11]}_k$. First, we derive $\gamma$ for
$s^{[11]}_k$. For $M=2$, let the eigenvalues and eigenvectors of
$\left(\mb{H}^{[11]}\right)^{-1}{\mb{H}^{[21]}}\left({\mb{H}^{[22]}}\right)^{-1}{\mb{H}^{[12]}}$
be $\lambda_1,\lambda_2$ and $\mb{u}_1,\mb{u}_2$, respectively. The
designs in \eqref{eq-bfstep2} can be expanded as
{\setlength{\arraycolsep}{2pt}\begin{align}\label{eq-bf}
\ol{\mb{v}}^{[11]}=\left[\begin{array}{cc}\mb{u}_1&\mb{u}_2\\\mb{u}_2&\mb{u}_1\\\mb{0}_2&\mb{0}_2
\end{array}\right],\ \ol{\mb{v}}^{[12]}=\left[\begin{array}{cc}\mb{u}_1&\mb{u}_2\\\mb{0}_2&\mb{0}_2\\\mb{u}_2&\mb{u}_1
\end{array}\right],
\end{align}}
where $\mb{0}_2$ denotes a $2\times 1$ zero vector. The eigenvalues
of
$\left(\ol{\mb{H}}^{[11]}\right)^{-1}{\ol{\mb{H}}^{[21]}}\left({\ol{\mb{H}}^{[22]}}\right)^{-1}{\ol{\mb{H}}^{[12]}}$
are arranged as
$\diag\left(\lambda_1,\lambda_2,\lambda_2,\lambda_2,\lambda_1,\lambda_1\right)$.
Inserting the designs of transmit beamformers in \eqref{eq-bfstep1}
into \eqref{eq-columneq} gives the received signals at Receiver $1$
\begin{align}\nonumber
\mb{y}^{[1]}&=\ol{\mb{H}}^{[11]}\ol{\mb{v}}^{[11]}\left[\begin{array}{c}s^{[11]}_1\\s^{[11]}_2\end{array}\right]+\ol{\mb{H}}^{[21]}\ol{\mb{v}}^{[21]}\left[\begin{array}{c}s^{[21]}_1\\s^{[21]}_2\end{array}\right]+\ol{\mb{H}}^{[11]}\ol{\mb{v}}^{[12]}\left[\begin{array}{c}s^{[12]}_1\\s^{[12]}_2\end{array}\right]+\ol{\mb{H}}^{[21]}\ol{\mb{v}}^{[22]}\left[\begin{array}{c}s^{[22]}_1\\s^{[22]}_2\end{array}\right]+\mb{w}^{[1]}\\
\nonumber
&=\ol{\mb{H}}^{[11]}\ol{\mb{v}}^{[11]}\left[\begin{array}{c}s^{[11]}_1\\s^{[11]}_2\end{array}\right]+\alpha^{[21]}\ol{\mc{H}}^{[21]}\ol{\mb{v}}^{[11]}\left[\begin{array}{c}s^{[21]}_1\\s^{[21]}_2\end{array}\right]+\ol{\mb{H}}^{[11]}\ol{\mb{v}}^{[12]}\left[\begin{array}{c}s^{[12]}_1\\s^{[12]}_2\end{array}\right]+\alpha^{[22]}\ol{\mb{H}}^{[11]}\ol{\mb{v}}^{[12]}\left[\begin{array}{c}s^{[22]}_1\\s^{[22]}_2\end{array}\right]+\mb{w}^{[1]}\\
\label{eq-sys}
&=\ol{\mb{H}}^{[11]}\ol{\mb{v}}^{[11]}\left[\begin{array}{c}s^{[11]}_1\\s^{[11]}_2\end{array}\right]+\alpha^{[21]}\ol{\mc{H}}^{[21]}\ol{\mb{v}}^{[11]}\left[\begin{array}{c}s^{[21]}_1\\s^{[21]}_2\end{array}\right]+\ol{\mb{H}}^{[11]}\ol{\mb{v}}^{[12]}\underset{\left[I_1\
I_2\right]^\t}{\underbrace{\left[\begin{array}{c}s^{[12]}_1+\alpha^{[22]}s^{[22]}_1\\s^{[12]}_2+\alpha^{[22]}s^{[22]}_2\end{array}\right]}}+\mb{w}^{[1]},
\end{align}
where
$\ol{\mc{H}}^{[21]}=\ol{\mb{H}}^{[21]}\left(\ol{\mb{H}}^{[22]}\right)^{-1}\ol{\mb{H}}^{[12]}$;
$I_1, I_2$ denote the aligned interference; and $\alpha^{[21]},
\alpha^{[22]}$ are coefficients to normalize the power of transmit
beamformers. Note that
$\left(\mb{H}^{[11]}\right)^{-1}\mb{H}^{[21]}\left(\mb{H}^{[22]}\right)^{-1}\mb{H}^{[12]}\mb{u}_i=\lambda_i\mb{u}_i$
due to the definition of eigenvalue decomposition. Let
${\mc{H}}^{[21]}={\mb{H}}^{[21]}\left({\mb{H}}^{[22]}\right)^{-1}{\mb{H}}^{[12]}$.
It follows ${\mc{H}}^{[21]}\mb{u}_i=\lambda_i\mb{H}^{[11]}\mb{u}_i$.

Replacing the designs in \eqref{eq-bf} into \eqref{eq-sys} gives
{\setlength{\arraycolsep}{1pt}\begin{align}\nonumber &\mb{y}^{[1]}=
\ol{\mb{H}}^{[11]}
\left[\begin{array}{cc}\mb{u}_1&\mb{u}_2\\\mb{u}_2&\mb{u}_1\\\mb{0}_2&\mb{0}_2
\end{array}\right]\left[\begin{array}{c}s^{[11]}_1\\s^{[11]}_2\end{array}\right]+\alpha^{[21]}(\mb{I}_3\otimes \mc{H}^{[21]})\left[\begin{array}{cc}\mb{u}_1&\mb{u}_2\\\mb{u}_2&\mb{u}_1\\\mb{0}_2&\mb{0}_2
\end{array}\right]\left[\begin{array}{c}s^{[21]}_1\\s^{[21]}_2\end{array}\right]+\ol{\mb{H}}^{[11]}\left[\begin{array}{cc}\mb{u}_1&\mb{u}_2\\\mb{0}_2&\mb{0}_2\\\mb{u}_2&\mb{u}_1
\end{array}\right]\left[\begin{array}{c}I_1\\I_2\end{array}\right]+\mb{w}^{[1]}\\
&=\ol{\mb{H}}^{[11]}\left[\begin{array}{cc}\mb{u}_1&\mb{u}_2\\\mb{u}_2&\mb{u}_1\\\mb{0}_2&\mb{0}_2
\end{array}\right]\left[\begin{array}{c}s^{[11]}_1\\s^{[11]}_2\end{array}\right]+\alpha^{[21]}\ol{\mb{H}}^{[11]}\left[\begin{array}{cc}\lambda_1\mb{u}_1&\lambda_2\mb{u}_2\\\lambda_2\mb{u}_2&\lambda_1\mb{u}_1\\\mb{0}_2&\mb{0}_2
\end{array}\right]\left[\begin{array}{c}s^{[21]}_1\\s^{[21]}_2\end{array}\right]+\ol{\mb{H}}^{[11]}\left[\begin{array}{cc}\mb{u}_1&\mb{u}_2\\\mb{0}_2&\mb{0}_2\\\mb{u}_2&\mb{u}_1
\end{array}\right]\left[\begin{array}{c}I_1\\I_2\end{array}\right]+\mb{w}^{[1]}.\label{eq-sys1}
\end{align}}
\hspace{-3pt}To decouple $s^{[11]}_1$, the receiver projects
$\mb{y}^{[1]}$ into the null of the subspaces spanned by the
equivalent channel vectors of
$s^{[11]}_2,s^{[21]}_1,s^{[21]}_2,I_1$, and $I_2$. The resulting
instantaneous normalized receive SNR $\gamma$ is upperbounded by
that of the scenario when projecting only the null of the subspace
spanned by $s^{[11]}_2,s^{[21]}_1$, and $s^{[21]}_2$. This
\emph{upperbound system} corresponds to the system equation without
aligned interference {\setlength{\arraycolsep}{2pt}\begin{align}
&\left[\begin{array}{c}\tilde{\mb{y}}_1^{[1]}\\\tilde{\mb{y}}_2^{[1]}\end{array}\right]=\left[\begin{array}{cc}{\mb{H}}^{[11]}&\\&\mb{H}^{[11]}\end{array}\right]\left[\begin{array}{cc}\mb{u}_1&\mb{u}_2\\\mb{u}_2&\mb{u}_1\end{array}\right]\left[\begin{array}{c}s^{[11]}_1\\s^{[11]}_2\end{array}\right]
+\alpha^{[21]}\left[\begin{array}{cc}\mb{H}^{[11]}&\\&\mb{H}^{[11]}\end{array}\right]\left[\begin{array}{cc}\lambda_1\mb{u}_1&\lambda_2\mb{u}_2\\\lambda_2\mb{u}_2&\lambda_1\mb{u}_1\end{array}\right]\left[\begin{array}{c}s^{[21]}_1\\s^{[21]}_2\end{array}\right]+\left[\begin{array}{c}\tilde{\mb{w}}^{[1]}_1\\\tilde{\mb{w}}^{[1]}_2\end{array}\right],\label{eq-sys2}
\end{align}}
\hspace{-3pt}where
$\left[\tilde{\mb{y}}^{[1]\t}_1,\tilde{\mb{y}}^{[1]\t}_2\right]^{\t}$
corresponds to the first four entries in $\mb{y}^{[1]}$ with
$\tilde{\mb{y}}^{[1]}_1,\tilde{\mb{y}}^{[1]}_2 \in
\mathds{C}^{2\times 1}$, and similar notations apply to
$\tilde{\mb{w}}^{[1]}_1,\tilde{\mb{w}}^{[1]}_2 \in
\mathds{C}^{2\times 1}$. To simplify the analysis, from Lemma
\ref{lemma2}, we can use a specific ZF receiver that does not lose
generality. We first invert the channel matrix $\mb{H}^{[11]}$ and
switch the positions of $s_2^{[11]}$ and $s_2^{[21]}$ as
\begin{align}\nonumber
\left[\begin{array}{c}\left({\mb{H}}^{[11]}\right)^{-1}\tilde{\mb{y}}_1^{[1]}\\\left({\mb{H}}^{[11]}\right)^{-1}\tilde{\mb{y}}_2^{[1]}\end{array}\right]&=\left[\begin{array}{cc}\mb{u}_1&\mb{u}_2\\\mb{u}_2&\kappa\mb{u}_1\end{array}\right]\left[\begin{array}{c}s^{[11]}_1\\\alpha^{[21]}\lambda_2
s^{[21]}_2\end{array}\right]\\
&+\left[\begin{array}{cc}\kappa\mb{u}_1&\mb{u}_2\\\mb{u}_2&\mb{u}_1\end{array}\right]\left[\begin{array}{c}\alpha^{[21]}\lambda_2s^{[21]}_1\\s^{[11]}_2\end{array}\right]+\left[\begin{array}{c}\left({\mb{H}}^{[11]}\right)^{-1}\tilde{\mb{w}}^{[1]}_1\\
\left({\mb{H}}^{[11]}\right)^{-1}\tilde{\mb{w}}^{[1]}_2\end{array}\right],\label{eq-kappa}
\end{align}
where $\kappa$ denotes the ratio of the eigenvalues of
$\left(\mb{H}^{[11]}\right)^{-1}{\mb{H}^{[21]}}\left({\mb{H}^{[22]}}\right)^{-1}{\mb{H}^{[12]}}$,
i.e., $\kappa=\frac{\lambda_1}{\lambda_2}$. Define
$\mb{u}=\left[\mb{u}_1\ \mb{u}_2\right]$ as the eigenvector matrix
of
$\left(\mb{H}^{[11]}\right)^{-1}{\mb{H}^{[21]}}\left({\mb{H}^{[22]}}\right)^{-1}{\mb{H}^{[12]}}$
and
\begin{align}
\mb{P}=\left[\begin{array}{cc}0&1\\1&0\end{array}\right],
\mb{Q}=\left[\begin{array}{cc}\kappa&0\\0&1\end{array}\right].
\end{align}
To cancel $s^{[21]}_1$ and $s^{[11]}_2$, the receiver calculates
$\tilde{\mb{y}}\in \mathds{C}^{2\times 1}$ as
\begin{align}\nonumber
&\tilde{\mb{y}}=(\mb{H}^{[11]})^{-1}\tilde{\mb{y}}^{[1]}_1-\mb{u}\mb{Q}(\mb{H}^{[11]}\mb{u}\mb{P})^{-1}\tilde{\mb{y}}^{[1]}_2=\left(\mb{u}-\mb{u}\mb{Q}\left(\mb{u}\mb{P}\right)^{-1}\mb{u}\mb{Q}\mb{P}\right)\left[\begin{array}{cc}s^{[11]}_1&\alpha^{[21]}\lambda_2
s^{[21]}_2\end{array}\right]^{\t}\\&+(\mb{H}^{[11]})^{-1}\tilde{\mb{w}}^{[1]}_1-\mb{u}(\mb{Q}\mb{P}^{-1})\mb{u}^{-1}\left(\mb{H}^{[11]}\right)^{-1}\tilde{\mb{w}}^{[1]}_2\nonumber\\
\label{eq-ZF3}&=(1-\kappa)\mb{u}\left[\begin{array}{cc}s^{[11]}_1&\alpha^{[21]}\lambda_2s^{[21]}_2\end{array}\right]^{\t}+(\mb{H}^{[11]})^{-1}\tilde{\mb{w}}^{[1]}_1-\mb{u}\mb{Q}\mb{P}\mb{u}^{-1}(\mb{H}^{[11]})^{-1}\tilde{\mb{w}}^{[1]}_2.
\end{align}
Note that the equivalent channel matrix is $\mb{u}$. To further
decouple $s^{[11]}_1$ from $s^{[21]}_2$ by ZF, the receiver
multiplies $\mb{u}^{-1}$ to the left side of $\tilde{\mb{y}}$ to
achieve
\begin{align}\label{eq-ZF4}
&\mb{u}^{-1}\tilde{\mb{y}}=(1-\kappa)\left[\begin{array}{cc}s^{[11]}_1&\alpha^{[21]}\lambda_2s^{[21]}_2\end{array}\right]^{\t}+\underset{\tilde{\mb{w}}}{\underbrace{\mb{u}^{-1}(\mb{H}^{[11]})^{-1}\tilde{\mb{w}}^{[1]}_1-\mb{Q}\mb{P}\mb{u}^{-1}(\mb{H}^{[11]})^{-1}\tilde{\mb{w}}^{[1]}_2}}.
\end{align}
Since the entries in $\tilde{\mb{w}}^{[1]}_1$ and
$\tilde{\mb{w}}^{[1]}_2$ are i.~i.~d.~$\mc{CN}(0,1)$ distributed,
the covariance matrix of the equivalent noise vector
$\tilde{\mb{w}}$ can be calculated as
\begin{align}
\mb{\Sigma}=\mb{u}^{-1}\left(\mb{H}^{[11]*}\mb{H}^{[11]}\right)^{-1}(\mb{u}^{-1})^*+\mb{Q}\mb{P}\mb{u}^{-1}\left(\mb{H}^{[11]*}\mb{H}^{[11]}\right)^{-1}(\mb{u}^{-1})^*\mb{P}\mb{Q}.
\end{align}
To decode $s^{[11]}_1$, the receiver uses the $(1,1)$th entry of
$\mb{\Sigma}$ as the variance for noise. Denote
\begin{align}\label{eq-delta}\mb{\Delta}=\left(\mb{u}^*\mb{H}^{[11]*}\mb{H}^{[11]}\mb{u}\right)^{-1}\end{align}
and its $(i,j)$th entry as $\delta_{ij}$. The noise variance in the
decoding of $s^{[11]}_1$ can be calculated as
$\delta_{11}+\kappa^2\delta_{22}$. The instantaneous normalized
receive SNR for this upperbound system can be expressed as
\begin{align}\label{eq-SNR2}\gamma'=\frac{(1-\kappa)^2}{\delta_{11}+\kappa^2\delta_{22}}.
\end{align}
Now, we focus on the outage probability of $\gamma$. Let $\epsilon$
be an arbitrary small positive number. The outage probability of the
upperbound system can be expanded as
\begin{align}\label{eq-cond}
P(\gamma<\epsilon)>P(\gamma'<\epsilon)>P(\gamma'<\epsilon|1\le|\kappa|\le
2)P(1\le|\kappa|\le 2).
\end{align}
Using the condition $1\le|\kappa|\le 2$, we can further upperbound
$\gamma'$ by
$\gamma'<\frac{(1+2)^2}{\delta_{11}+\delta_{22}}=\frac{9}{\tr
\mb{\Delta}}$. Thus, we have
\begin{align*}
P(\gamma<\epsilon)>P\left(\frac{1}{\tr
\mb{\Delta}}<\frac{\epsilon}{9}\right)P(1\le|\kappa|\le 2).
\end{align*}
Recall that $\kappa$ is the ratio of the eigenvalues of
$\left(\mb{H}^{[11]}\right)^{-1}{\mb{H}^{[21]}}\left({\mb{H}^{[22]}}\right)^{-1}{\mb{H}^{[12]}}$.
All channel matrices are independently generated from a continuous
distribution. Thus, $P(1\le|\kappa|\le 2)$ is a bounded nonzero
positive number. It suffices to rely on the scaling of the outage
probability of $\frac{1}{\tr \mb{\Delta}}$. From the definition of
$\mb{\Delta}$ in \eqref{eq-delta}, we have
\begin{align*}
\frac{1}{\tr \mb{\Delta}}=\frac{1}{\tr \left(
\left(\mb{u}^*\mb{H}^{[11]*}\mb{H}^{[11]}\mb{u}\right)^{-1}\right)}=
\frac{1}{\tr
\left((\mb{H}^{[11]*}\mb{H}^{[11]})^{-1}\left(\mb{u}\mb{u}^*\right)^{-1}\right)}<\frac{2}{\tr
\left((\mb{H}^{[11]*}\mb{H}^{[11]})^{-1}\right)}.
\end{align*}
The inequality in the last line is valid because
$\mb{u}\mb{u}^{*}\prec\tr(\mb{u}\mb{u}^{*})\mb{I}_2=2\mb{I}_2$,
where $\tr(\mb{u}\mb{u}^{*})\mb{I}_2-\mb{u}\mb{u}^{*}$ is a positive
definite matrix. Applying Lemma \ref{lemma1} to the term
$\frac{1}{\tr
\left((\mb{H}^{[11]})^{-1}(\mb{H}^{[11]*})^{-1}\right)}$ results in
a diversity gain of only 1. Thus, the achievable diversity for
$\frac{1}{\tr \mb{\Delta}}$ is not larger than 1. This concludes the
proof.

\section{Diversity analysis for the modified JaSh scheme}\label{appen2}
The modified JaSh scheme collects two alignment blocks and uses
Alamouti codes as the inner codes. The resulting instantaneous
normalized receive SNR is the sum of those of $s_1^{[11]}$ and
$s_2^{[11]}$ in the JaSh scheme. In this appendix, we present the
analysis for the modified JaSh scheme.
\begin{theorem}\label{thm-modi}
In the short-term regime, the achievable diversity gain of the
modified JaSh scheme is no more than 1 for the $2\times 2$
double-antenna X channel.
\end{theorem}
\begin{proof}
The analysis is similar to the proof of Theorem \ref{thm-prediv}.
From \eqref{eq-sys1}, the instantaneous normalized receive SNR of
$s_2^{[11]}$ can be obtained from that of $s_1^{[11]}$ by swapping
$\mb{u}_1$ and $\mb{u}_2$. Similar to the specific ZF receiver in
\eqref{eq-ZF3} and \eqref{eq-ZF4}, we can obtain an upperbound on
the instantaneous normalized receive SNR of $s_2^{[11]}$ from
\eqref{eq-SNR2} as
\begin{align*}
\gamma_2^{[11]}<\frac{(1-\kappa)^2}{\kappa^2\delta_{11}+\delta_{22}},
\end{align*}
where $\kappa$ and $\delta_{ij}$ are defined in \eqref{eq-kappa} and
\eqref{eq-delta}, respectively. Since the use of Alamouti codes
accumulates the SNRs of $s_1^{[11]}$ and $s_2^{[11]}$, we have
\begin{align*}
\gamma_1^{[11]}+\gamma_2^{[11]}<\frac{(1-\kappa)^2}{\delta_{11}+\kappa^2\delta_{22}}+\frac{(1-\kappa)^2}{\kappa^2\delta_{11}+\delta_{22}}.
\end{align*}
The conditional bounding technique in \eqref{eq-cond} can be
straightforwardly applied as
\begin{align*}
&P\left(\gamma_1^{[11]}+\gamma_2^{[11]}<\epsilon\right)>P\left(\frac{(1-\kappa)^2}{\delta_{11}+\kappa^2\delta_{22}}+\frac{(1-\kappa)^2}{\kappa^2\delta_{11}+\delta_{22}}<\epsilon\mid1\le|\kappa|\le
2\right)P(1\le|\kappa|\le
2)\\
&>P\left(\frac{1}{\delta_{11}+\delta_{22}}<\frac{\epsilon}{18}\right)P(1\le|\kappa|\le
2).
\end{align*}
The rest of the proof is similar to that of Theorem \ref{thm-prediv}
by showing that the scaling of the outage probability of
$P\left(\frac{1}{\delta_{11}+\delta_{22}}<\epsilon\right)$ has only
diversity 1. This concludes the proof.
\end{proof}
The results of Theorem \ref{thm-modi} are surprising. Although the
instantaneous normalized receive SNRs $\gamma_1^{[11]}$ and
$\gamma_2^{[11]}$ are correlated, they are still distinct. Theorem
\ref{thm-modi} implies that the sum of two distinct SNRs is not
sufficient to achieve a diversity of 2.

\section{Proof of Theorem \ref{thm-div}}\label{appen1}
The proof is based on the outage probability of the instantaneous
normalized receive SNR of $s_1^{[11]}$. Since the design is
symmetric for all symbols, similar diversity results apply to the
decoding of other symbols. Let
$\hat{\mb{H}}^{[21]}=\left[\frac{\hat{\mb{H}}_1^{[21]*}}{\left\|\hat{\mb{H}}^{[21]}_1\right\|^2}\
-\frac{\hat{\mb{H}}_2^{[21]*}}{\left\|\hat{\mb{H}}_2^{[21]}\right\|^2}\right]^*$,
$\hat{\mb{H}}^{[11]}=\left[\hat{\mb{H}}_1^{[11]*}\
\hat{\mb{H}}_2^{[11]*}\right]^*$. The equivalent system in
\eqref{eq-remainedeq} can be rewritten as
\begin{align}\hat{\mb{y}}=\sqrt{\frac{3}{4}}\hat{\mb{H}}^{[21]*}\hat{\mb{H}}^{[11]}\left[\begin{array}{c}s^{[11]}_1\\s^{[11]}_2\end{array}\right]+\hat{\mb{H}}^{[21]*}\left[\begin{array}{c}\hat{\mb{w}}_1\\
\hat{\mb{w}}_2\end{array}\right].\end{align} The covariance matrix
of the equivalent noise vector is
$\hat{\mb{H}}^{[21]*}\mb{\Sigma}_{\hat{\mb{w}}}\hat{\mb{H}}^{[21]}$,
where $\mb{\Sigma}_{\hat{\mb{w}}}=\diag\left(1,2,1,2\right).$ Let
the first column of $\hat{\mb{H}}^{[11]}$ be
$\hat{\mb{h}}_1^{[11]}$, where
$\hat{\mb{h}}_1^{[11]}=\left[\tilde{h}_{11}^{[11]}\
\tilde{h}_{21}^{[11]*}\ \tilde{h}_{12}^{[11]}\
\tilde{h}_{22}^{[11]*}\right]^{\t}$. The instantaneous normalized
receive SNR of $s_1^{[11]}$ can be expressed as
\begin{align}\nonumber
\gamma=\frac{3}{4}\left(\hat{\mb{H}}^{[21]*}\hat{\mb{h}}^{[11]}_1\right)^*\left(\hat{\mb{H}}^{[21]*}\mb{\Sigma}_{\hat{\mb{w}}}\hat{\mb{H}}^{[21]}\right)^{-1}\hat{\mb{H}}^{[21]*}\hat{\mb{h}}^{[11]}_1.\end{align}
Define
$\bar{\gamma}=\hat{\mb{h}}_1^{[11]*}\hat{\mb{H}}^{[21]}\left(\hat{\mb{H}}^{[21]*}\hat{\mb{H}}^{[21]}\right)^{-1}\hat{\mb{H}}^{[21]*}\hat{\mb{h}}^{[11]}_1$.
It can be shown that $\frac{3}{4}\bar{\gamma}\ge \gamma \ge
\frac{3}{8} \bar{\gamma}$. By \eqref{eq-div}, $\gamma$ and
$\bar{\gamma}$ have the same diversity. Thus, we focus on analyzing
the outage probability of $\bar{\gamma}$ to get rid of
$\mb{\Sigma}_{\hat{\mb{w}}}$. Since the columns of
$\hat{\mb{H}}^{[21]}$ are orthogonal, $\bar{\gamma}$ can be further
simplified as
\begin{align}\label{eq-gamma}\bar{\gamma}=\underset{b^{[21]}}{\underbrace{\left(\frac{1}{2\left\|\hat{\mb{H}}^{[21]}_1\right\|^2}+\frac{1}{2\left\|\hat{\mb{H}}^{[21]}_2\right\|^2}\right)^{-1}}}\hat{\mb{h}}_1^{[11]*}\hat{\mb{H}}^{[21]}\hat{\mb{H}}^{[21]*}\hat{\mb{h}}^{[11]}_1.\end{align}
It is complicated to analyze the distribution of $\bar{\gamma}$
directly. Instead, we fix $\mb{H}^{[21]}$, $\mb{H}^{[12]}$, and
$\mb{H}^{[22]}$, and allow only $\mb{H}^{[11]}$ to change. Then,
$\hat{\mb{H}}^{[21]}$ is fixed, whereas $\hat{\mb{h}}_1^{[11]}$ is
still a random vector. Since $\tilde{\mb{H}}^{[11]}=\frac{1}{
\left\|(\mb{H}^{[12]})^{-1}\right\|}(\mb{H}^{[12]})^{-1}\mb{H}^{[11]}$,
it can be shown that the conditional distribution of $\bar{\gamma}$
is a generalized Chi-square distribution with degree 2. The
covariance matrix of the components in the generalized Chi-square
distribution can be calculated as
\begin{align}\label{eq-covm}\mb{\Phi}=\underset{\mb{H}^{[11]}|\mb{H}^{[ji]},
(j,i)\neq (1,1)}{\Exp}
b^{[21]}\hat{\mb{H}}^{[21]*}\hat{\mb{h}}^{[11]}_1\hat{\mb{h}}^{[11]*}_1\hat{\mb{H}}^{[21]}=b^{[21]}\hat{\mb{H}}^{[21]*}\left(\underset{\mb{H}^{[11]}|\mb{H}^{[ji]},
(j,i)\neq (1,1)}{\Exp}\left(\hat{\mb{h}}^{[11]}_1\hat{\mb{h}}_1^{[11]*}\right)\right)\hat{\mb{H}}^{[21]}.%\\
%&=\left(\frac{1}{\|\hat{\mb{G}}_1\|^2}+\frac{1}{\|\hat{\mb{G}}_2\|^2}\right)^{-1}\hat{\mb{G}}^*\left[\begin{array}{cc}\Exp
%\left[\begin{array}{c}\tilde{h}_{11}\\ \tilde{h}_{21}^*\end{array}\right]\left[\begin{array}{cc}\tilde{h}_{11}^* \tilde{h}_{21}\end{array}\right] &\mb{0}\\
%\mb{0}& \Exp \left[\begin{array}{c}\tilde{h}_{12}\\
%\tilde{h}_{22}^*\end{array}\right]\left[\begin{array}{cc}\tilde{h}_{12}^*
%\tilde{h}_{22}\end{array}\right]\end{array}\right]\hat{\mb{G}}.
\end{align}
The equality holds because $\hat{\mb{H}}^{[21]}$ only depends on
$\tilde{\mb{H}}^{[21]}$, which depends on $\mb{H}^{[21]}$ and
$\mb{H}^{[22]}$. Thus, $\hat{\mb{H}}^{[21]}$ is independent from
$\mb{H}^{[11]}$. It can be calculated that $\Exp
\left(\tilde{h}^{[11]}_{1i}\tilde{h}_{1i}^{[11]*}\right)=\frac{\left\|\left(\mb{H}^{[12]}\right)^{-1}_{1\cdot}\right\|^2}{\left\|\left(\mb{H}^{[12]}\right)^{-1}\right\|^2}$,
$\Exp
\left(\tilde{h}^{[11]}_{2i}\tilde{h}_{2i}^{[11]*}\right)=\frac{\left\|\left(\mb{H}^{[12]}\right)^{-1}_{2\cdot}\right\|^2}{\left\|\left(\mb{H}^{[12]}\right)^{-1}\right\|^2}$,
and $\Exp\left(\tilde{h}^{[11]}_{1i}\tilde{h}^{[11]}_{2i}\right)=0$
for $i=1,2$, where $\left(\mb{H}^{[12]}\right)^{-1}_{k\cdot}$
denotes the $k$th row of $\left(\mb{H}^{[12]}\right)^{-1}$. Let
$\mb{\Theta}_{\mb{H}^{[12]}}=\diag\left(\frac{\left\|\left(\mb{H}^{[12]}\right)^{-1}_{1\cdot}\right\|^2}{\left\|\left(\mb{H}^{[12]}\right)^{-1}\right\|^2},
\frac{\left\|\left(\mb{H}^{[12]}\right)^{-1}_{2\cdot}\right\|^2}{\left\|\left(\mb{H}^{[12]}\right)^{-1}\right\|^2}\right)$.
The covariance matrix can be simplified as
\begin{align}\label{eq-cov}
\mb{\Phi}=b^{[21]}\left(\frac{\hat{\mb{H}}^{[21]*}_1\mb{\Theta}_{\mb{H}^{[12]}}\hat{\mb{H}}^{[21]}_1}{\left\|\hat{\mb{H}}^{[21]}_1\right\|^4}+\frac{\hat{\mb{H}}_2^{[21]*}\mb{\Theta}_{\mb{H}^{[12]}}\hat{\mb{H}}_2^{[21]}}{\left\|\hat{\mb{H}}^{[21]}_2\right\|^4}\right).\end{align}

Given the covariance matrix, we calculate the outage probability of
$\bar{\gamma}$ conditioned on $\mb{H}^{[21]}$, $\mb{H}^{[12]}$, and
$\mb{H}^{[22]}$. Denote the eigenvalues of $\mb{\Phi}$ as
$\lambda_1$ and $\lambda_2$. Since the distribution of
$\bar{\gamma}$ is a generalized Chi-square with degree 2, the
probability density function (pdf) of $\bar{\gamma}$ is
$f_{\bar{\gamma}}=\frac{\exp\left(-\bar{\gamma}/\lambda_1\right)}{\lambda_1-\lambda_2}+\frac{\exp\left(-\bar{\gamma}/\lambda_2\right)}{\lambda_2-\lambda_1}$.
It follows that \begin{align*}
&P\left(\bar{\gamma}<\epsilon|\mb{H}^{[21]},\mb{H}^{[12]},\mb{H}^{[22]}\right)=\int_{0}^\epsilon\frac{\exp\left(-\frac{\bar{\gamma}}{\lambda_1}\right)}{\lambda_1-\lambda_2}+\frac{\exp\left(-\frac{\bar{\gamma}}{\lambda_2}\right)}{\lambda_2-\lambda_1}d\bar{\lambda}\\
&=\frac{\lambda_1}{\lambda_1-\lambda_2}\left(1-\exp\left(-\frac{\epsilon}{\lambda_1}\right)\right)-\frac{\lambda_2}{\lambda_1-\lambda_2}\left(1-\exp\left(-\frac{\epsilon}{\lambda_2}\right)\right)\\
&=1-\frac{\lambda_1}{\lambda_1-\lambda_2}\exp\left(-\frac{\epsilon}{\lambda_1}\right)+\frac{\lambda_2}{\lambda_1-\lambda_2}\exp\left(-\frac{\epsilon}{\lambda_2}\right)\\
&=1-\frac{\lambda_1}{\lambda_1-\lambda_2}\left(1-\frac{\epsilon}{\lambda_1}+\frac{\epsilon^2}{\lambda_1^2}\right)+\frac{\lambda_2}{\lambda_1-\lambda_2}\left(1-\frac{\epsilon}{\lambda_2}+\frac{\epsilon^2}{\lambda_2^2}\right)+o\left(\epsilon^2\right)=\frac{\epsilon^2}{\lambda_1\lambda_2}+o\left(\epsilon^2\right).
\end{align*}

Using \eqref{eq-div}, the diversity gain of $\bar{\gamma}$ can be
calculated as \begin{align*}& d=\lim_{\epsilon\rightarrow
0}\frac{\log P(\gamma<\epsilon)}{\log
\epsilon}=\lim_{\epsilon\rightarrow 0}\frac{\log
\underset{\mb{H}^{[ji]},(j,i)\neq(1,1)}{\Exp}P\left(\gamma<\epsilon|\mb{H}^{[ji]},(j,i)\neq(1,1)\right)}{\log
\epsilon}\\
&=\lim_{\epsilon\rightarrow 0}\frac{\log
\epsilon^2\left(\underset{\mb{H}^{[ji]}}{\Exp}\frac{1}{\lambda_1\lambda_2}\right)+o(\epsilon^2)}{\log
\epsilon}.\end{align*} Obviously, the achievable diversity gain is 2
if and only if
$\underset{\mb{H}^{[ji]}}{\Exp}\frac{1}{\lambda_1\lambda_2}$ is
bounded by a limited number. Next, we show that
$\underset{\mb{H}^{[ji]}}{\Exp}\frac{1}{\lambda_1\lambda_2}$ is
upperbounded by a limited number, followed by being lowerbounded by
another number.

%Denote
%$\beta=\min\left\{\frac{\left(\mb{A}^{-1}\right)_{1\cdot}\left(\mb{A}^{-1}\right)_{1\cdot}^*}{\tr
%(\mb{A}^{-1}\mb{A}^{-1*})},
%\frac{\left(\mb{A}^{-1}\right)_{2\cdot}\left(\mb{A}^{-1}\right)_{2\cdot}^*}{\tr
%(\mb{A}^{-1}\mb{A}^{-1*})}\right\}$. From \eqref{eq-cov}, it follows
%that \be \mb{\Phi}\succ
%\left(\frac{1}{\|\hat{\mb{G}}_1\|^2}+\frac{1}{\|\hat{\mb{G}}_2\|^2}\right)^{-1}\left(\frac{2\beta\hat{\mb{G}}_1^*\hat{\mb{G}}_1}{\|\hat{\mb{G}}_1\|^4}+\frac{2\beta\hat{\mb{G}}_2^*\hat{\mb{G}}_2}{\|\hat{\mb{G}}_2\|^4}\right)=2\beta
%\mb{I}_2.\ee Note that $\lambda_1$ and $\lambda_2$ are the
%eigenvalues of $\mb{\Phi}$. Then, $\lambda_1\lambda_2>4\beta^2$. We
%have \be
%\underset{\mb{A},\mb{B},\mb{G}}{\Exp}\frac{1}{\lambda_1\lambda_2}<\frac{1}{4}\underset{\mb{A}}{\Exp}\frac{1}{\beta^2}\ee

Theorem VI.7.1 in \cite{matrixanalysis} introduces a lowerbound on
the determinant of the sum of two Hermitian matrices. Since
$\mb{\Theta}_{\mb{H}^{[12]}}$ is a diagonal matrix, we have $\det
\mb{\Phi}> \det \mb{\Theta}_{\mb{H}^{[12]}}$. It follows,
\begin{align}\label{eq-bound1}
\underset{\mb{H}^{[ji]}}{\Exp}\frac{1}{\lambda_1\lambda_2}=\underset{\mb{H}^{[ji]}}{\Exp}\frac{1}{\det\mb{\Phi}}<\underset{\mb{H}^{[12]}}{\Exp}\frac{1}{\det\mb{\Theta}_{\mb{H}^{[12]}}}<\underset{\mb{H}^{[12]}}{\Exp}\frac{1}{\det\mb{M}_{\mb{H}^{[12]}}},
\end{align}
where
$\mb{M}_{\mb{H}^{[12]}}=\frac{(\mb{H}^{[12]})^{-1}(\mb{H}^{[12]*})^{-1}}{\left\|\left(\mb{H}^{[12]}\right)^{-1}\right\|^2}$.
The last inequality is valid because of the Hadamard inequality
since the diagonal entries of $\mb{\Theta}_{\mb{H}^{[12]}}$ and
$\mb{M}_{\mb{H}^{[12]}}$ are the same. Let the eigenvalues of
$\mb{H}^{[12]}\mb{H}^{[12]*}$ be $x_1$ and $x_2$, whose joint pdf
can be expressed as
$\frac{1}{2\pi}\exp\left(-\frac{x_1^2+x_2^2}{2}\right)(x_1-x_2)^2$.
The RHS of \eqref{eq-bound1} can be calculated as \begin{align*}
&\underset{\mb{H}^{[12]}}{\Exp}\frac{1}{\det\mb{M}_{\mb{H}^{[12]}}}=\underset{x_1,x_2}{\Exp}
\frac{\left(x_1^{-1}+x_2^{-1}\right)^2}{x_1^{-1}x_2^{-1}}=\underset{x_1,x_2}{\Exp}
\left(2+\frac{x_2}{x_1}+\frac{x_1}{x_2}\right)=2+2\underset{x_1,x_2}{\Exp}\frac{x_2}{x_1}\\
&\underset{x_1,x_2}{\Exp}\frac{x_2}{x_1}=\frac{1}{2\pi}\int
\frac{x_2}{x_1}\exp
\left(-\frac{x_1^2+x_2^2}{2}\right)(x_1-x_2)^2dx_1dx_2<\frac{1}{2\pi}\int
x_1x_2\exp
\left(-\frac{x_1^2+x_2^2}{2}\right)dx_1dx_2=\frac{2}{\pi}.
\end{align*}
The last inequality holds because $(x_1-x_2)^2<x_1^2$. Thus, we have
shown that
$\underset{\mb{H}^{[ji]}}{\Exp}\frac{1}{\lambda_1\lambda_2}$ is
upperbounded by $2+\frac{4}{\pi}$. Finally, we show the lowerbound.
Since the sum of the diagonal entries in
$\mb{\Theta}_{\mb{H}^{[12]}}$ is equal to 1, i.e.,
$\frac{\left\|\left(\mb{H}^{[12]}\right)^{-1}_{1\cdot}\right\|^2}{\left\|\left(\mb{H}^{[12]}\right)^{-1}\right\|^2}+\frac{\left\|\left(\mb{H}^{[12]}\right)^{-1}_{2\cdot}\right\|^2}{\left\|\left(\mb{H}^{[12]}\right)^{-1}\right\|^2}=1$,
we have $\mb{\Theta}_{\mb{H}^{[12]}}\prec \mb{I}_2$. Then, from
\eqref{eq-cov},$\mb{\Phi}\prec
b^{[21]}\left(\frac{\hat{\mb{H}}_1^{[21]*}\hat{\mb{H}}^{[21]}_1}{\left\|\hat{\mb{H}}^{[21]}_1\right\|^4}+\frac{\hat{\mb{H}}_2^{[21]*}\hat{\mb{H}}^{[21]}_2}{\left\|\hat{\mb{H}}^{[21]}_2\right\|^4}\right)=
\mb{I}_2$. It follows,
\begin{align*}\underset{\mb{H}^{[ji]}}{\Exp}\frac{1}{\lambda_1\lambda_2}=\underset{\mb{H}^{[ji]}}{\Exp}\frac{1}{\det\mb{\Phi}}>\frac{1}{\det
\mb{I}_2}=1. \end{align*} Therefore,
$\underset{\mb{H}^{[ji]}}{\Exp}\frac{1}{\lambda_1\lambda_2}$ is
lowerbounded by $1$.

\newpage

%\begin{figure}
%  % Requires \usepackage{graphicx}
%  \centering
%  \includegraphics[width=4in]{images/X}\\
%  \caption{The $2\times 2$ $M$-antenna MIMO X Channels.}\label{fig-X}
%\end{figure}

\begin{figure}\centering
  \includegraphics[width=5in]{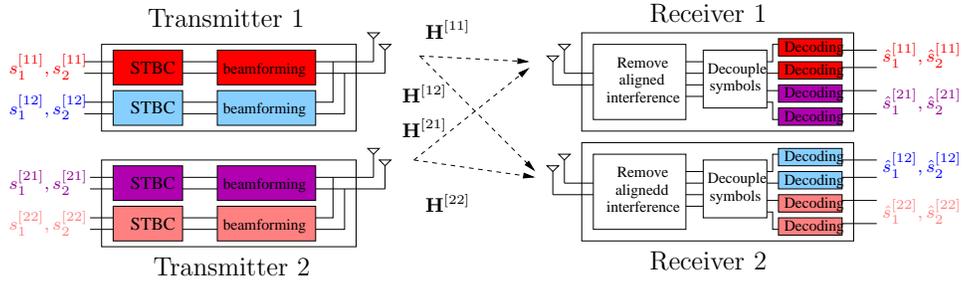}\\
  \caption{IA designs using Alamouti codes for the $2\times 2$ double-antenna X channel. }\label{fig-diag}
\end{figure}

\begin{figure}\centering
  \includegraphics[width=4in]{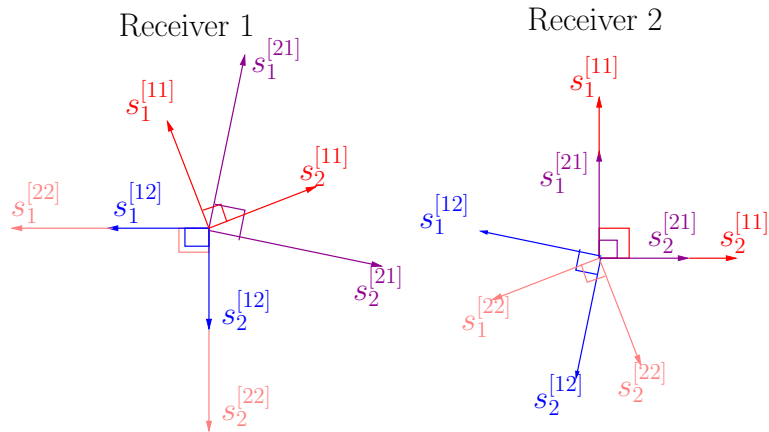}\\
  \caption{Receiver signal space using the Alamouti designs. Each arrow represent one equivalent channel vector. Each receiver observes a $6$-dimensional signal space. Two dimensions are for aligned interference, and the remained four dimensions are for desired symbols. The equivalent channel vector of $s_1^{[ji]}$ is orthogonal to that of $s_2^{[ji]}$.}\label{fig-space}
\end{figure}

\begin{figure}
\centering
  \includegraphics[width=5in]{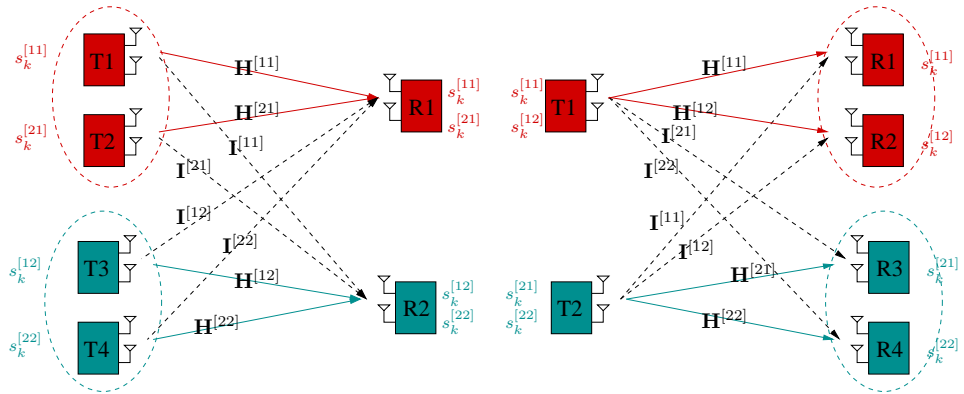}\\
  \caption{Network models of the two-cell IMAC and the IBC. In each cell, one BS is serving two users. Transmission in one cell creates interference to the other cell. The desired links are represented by solid lines, whereas interferring links are represented by dashed lines. }\label{fig-InterBC}
\end{figure}

\begin{figure}
\centering
  \includegraphics[width=4in]{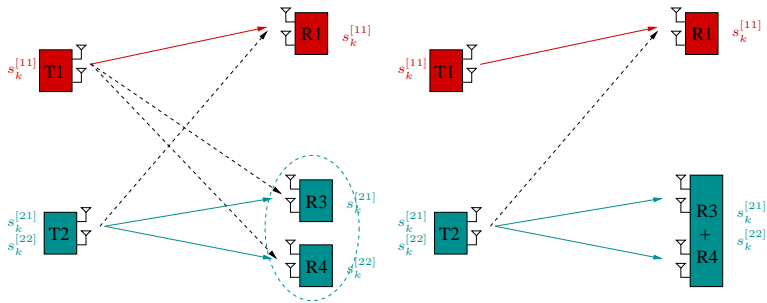}\\
  \caption{Proof of the outerbound on the DoF region of the IBC: a modified IBC (left side) and a Z channel (right side).}\label{fig-IBCOuterbound}
\end{figure}

\begin{figure}
\centering
  \includegraphics[width=4in]{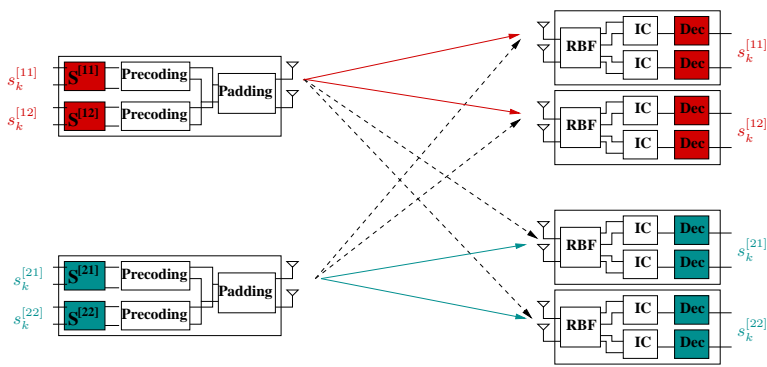}\\
  \caption{System diagram of the proposed transmission method in the two-cell IBC.}\label{fig-IBCStr0}
\end{figure}

\begin{figure}
\centering
  \includegraphics[width=3.5in]{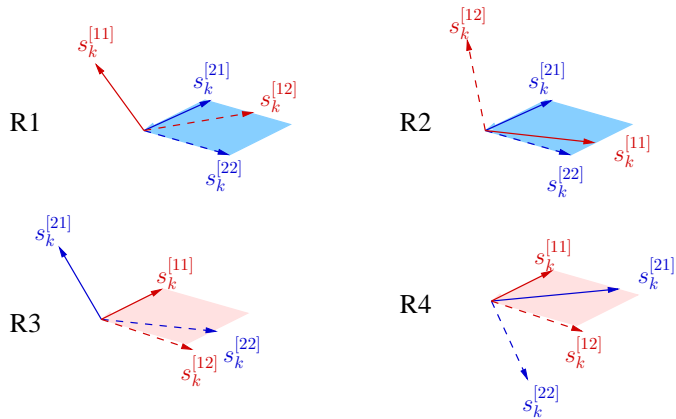}\\
  \caption{Alignment in the six-dimensional receive signal space. Each arrow represents two dimensions carrying $s_k^{[ji]},\ k\in \{1,2\}$. Six interferring symbols are aligned in four-dimensional subspace. The remaining two dimensions are for desired symbols.}\label{fig-align}
\end{figure}

\begin{figure}
\centering
\includegraphics[height=4in,angle=-90]{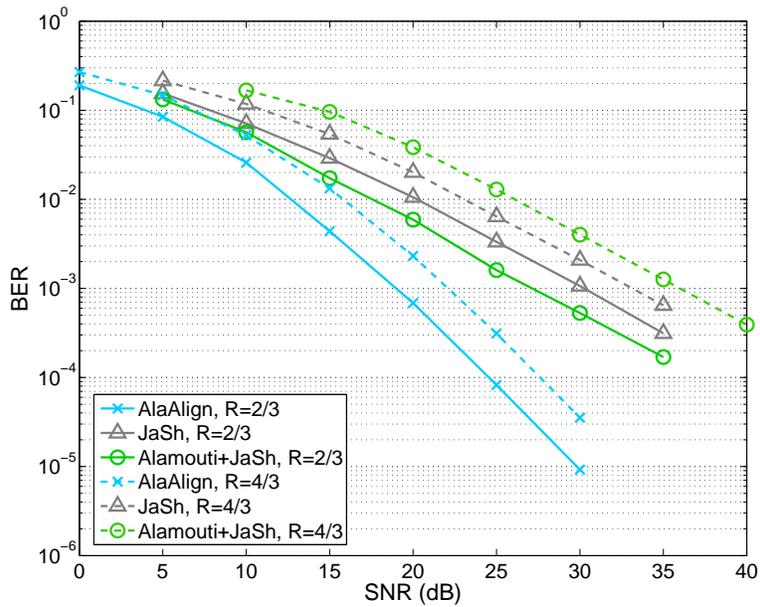}
\caption{BER comparison in the X channels for the proposed method
(labeled as `AlaAlign'), the JaSh scheme, and the modified JaSh
scheme (labeled as `Alamouti+JaSh'). Rate $R$ is measured as bits
per channel use per node pair. } \label{fig-XchanNew}
\end{figure}

\begin{figure}
\centering
  \includegraphics[width=4in]{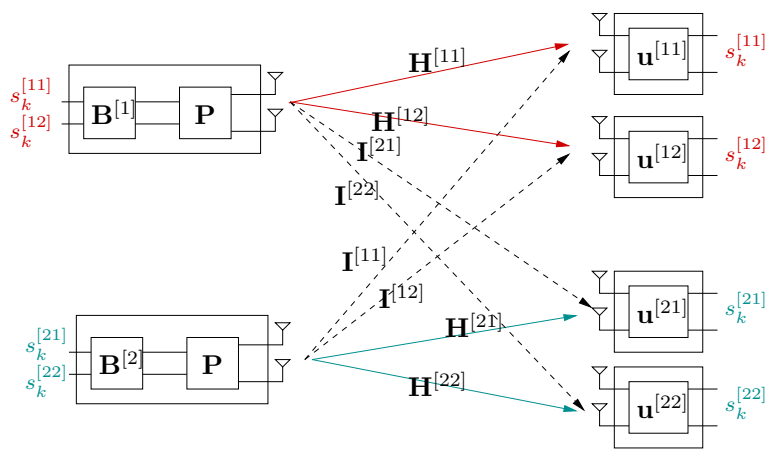}\\
  \caption{System diagram for the extended downlink interference alignment \cite{Suh11} in the two-user double-antenna IBC.}\label{fig-IBCStr}
\end{figure}

\begin{figure}
\centering
\includegraphics[height=4in,angle=-90]{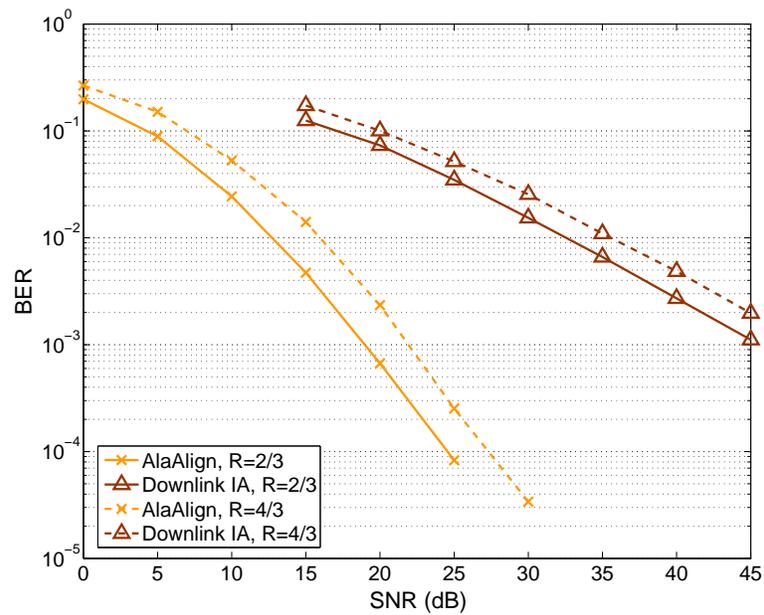}
\caption{BER comparison in the IBC between the proposed method
(labeled as `AlaAlign') and the downlink IA method (labeled as
`Downlink IA'). Rate $R$ is measured as bits per channel use per
receiver.} \label{fig-IBCNew}
\end{figure}

\begin{figure}
\centering
\includegraphics[height=4in,angle=-90]{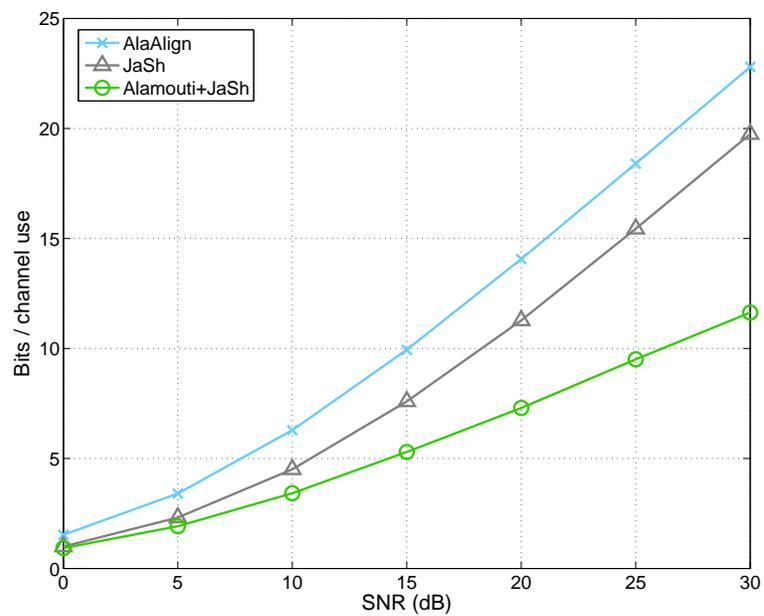}
\caption{Achievable ergodic mutual information in the two-user X
channel.} \label{fig-XCap}
\end{figure}

\begin{figure}
\centering
\includegraphics[height=4in,angle=-90]{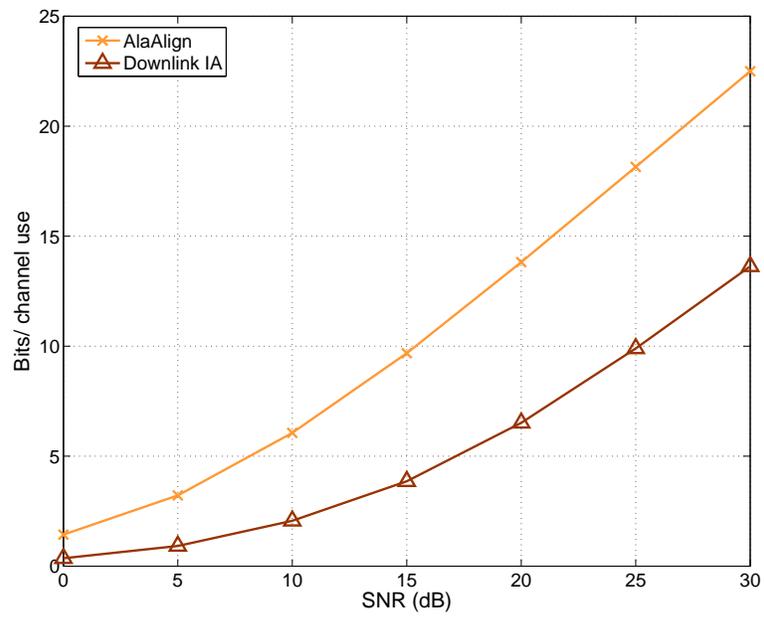}
\caption{Achievable ergodic mutual information in the two-cell IBC.}
\label{fig-IBCCap}
\end{figure}

\end{document}